\PassOptionsToPackage{
  colorlinks=true,
  linkcolor=RoyalBlue,
  citecolor=RoyalBlue,
  urlcolor=RoyalBlue
}{hyperref}
\documentclass[a4paper,UKenglish, autoref, thm-restate]{lipics-v2021}
\nolinenumbers
\hideLIPIcs
\usepackage[ruled,noend,linesnumbered]{algorithm2e} %
\usepackage{setspace} %
\DontPrintSemicolon     %
\SetNlSty{}{}{}                %
\SetAlgoInsideSkip{smallskip}   %
\SetAlFnt{\small}			%
\SetAlCapFnt{\small}		%
\SetAlCapNameFnt{\small}
\SetArgSty{textnormal}
\newcommand{\algocomment}[1]{\textcolor{teal}{{//#1}}}
\newcommand{\pluseq}{\mathrel{+}=}

\usepackage{graphicx} %
\usepackage{amsmath, amssymb, amsthm}
\usepackage{mathtools}
\usepackage{todonotes}
\usepackage{bbm}

\usepackage[capitalize,noabbrev,nameinlink]{cleveref}
\crefname{observation}{Observation}{Observations}
\crefname{claim}{Claim}{Claims}
\usepackage[dvipsnames,svgnames]{xcolor}
\usepackage{enumitem}
\usepackage{cite}
\usepackage[table]{xcolor}
\usepackage{multirow}
\usepackage{apxproof}
\usepackage{slashbox}
\usepackage{tikz}
\usetikzlibrary{arrows.meta}

\makeatletter
\def\Cline#1#2{\@Cline#1#2\@nil}
\def\@Cline#1-#2#3\@nil{%
  \omit
  \@multicnt#1%
  \advance\@multispan\m@ne
  \ifnum\@multicnt=\@ne\@firstofone{&\omit}\fi
  \@multicnt#2%
  \advance\@multicnt-#1%
  \advance\@multispan\@ne
  \leaders\hrule\@height#3\hfill
  \cr}
\makeatother

\newcommand{\christoph}[1]{{{\color{Maroon}{[\textbf{Christoph}: #1]}}}}
\newcommand{\benny}[1]{{{\color{red}{[\textbf{Benny}: #1]}}}}

\newcommand{\hide}[1]{} %

\def\fpsharpp{FP${^{\mathrm{\#P}}}$\xspace}

\def\cellcolorblue{\cellcolor{SkyBlue!30}}

\colorlet{addedcolor}{blue!50!gray}
\colorlet{changedcolor}{purple}
\colorlet{finalchangecolor}{green}
\AtBeginDocument{\colorlet{defaultcolor}{.}}
\newcommand{\added}[1]{{\color{defaultcolor}{}\ignorespaces#1}}

\def\one{\mathbbm{1}}

\def\mytitle{The Importance of Parameters in Ranking Functions}

\author{Christoph Standke}{RWTH Aachen University, Aachen, Germany}{standke@informatik.rwth-aachen.de}{https://orcid.org/0000-0002-9614-05040000-0002-3034-730X}{\href{https://cordis.europa.eu/project/id/101054974}{ERC grant  101054974 (SymSim)} and \href{https://gepris.dfg.de/gepris/projekt/282652900?language=en}{German Research Foundation grant GRK 2236 (UnRAVeL)}.
Views and opinions expressed are however those of the author(s) only and do not necessarily reflect
those of the European Union or the European Research Council. Neither the European Union nor
the granting authority can be held responsible for them.}

\author{Nikolaos Tziavelis}{UC Santa Cruz, Santa Cruz, California, USA}{ntziavel@ucsc.edu}{https://orcid.org/0000-0002-9614-05040000-0001-8342-2177}{}

\author{Wolfgang Gatterbauer}{Northeastern University, Boston, Massachusetts, USA}{w.gatterbauer@northeastern.edu}{https://orcid.org/0000-0002-9614-0504}{\href{https://www.nsf.gov/awardsearch/show-award/?AWD_ID=1762268}{NSF Career Award IIS-1762268}}

\author{Benny Kimelfeld}{Technion, Haifa, Israel and RelationalAI Inc., Berkeley, California, USA}{bennyk@technion.ac.il}{https://orcid.org/0000-0002-9614-05040000-0002-7156-1572}{German Research Foundation grant KI 2348/1-1}

\authorrunning{Standke, Tziavelis, Gatterbauer, Kimelfeld}

\title{\mytitle}

\titlerunning{\mytitle}

\def\vect#1{\mathbf{#1}}

\def\e#1{\emph{#1}}
\def\piid{\pi_{\mathsf{id}}}

\newcommand{\defeq}{\coloneqq}
\def\angs#1{\mathord{\langle #1 \rangle}}

\def\scoref#1{\mathsf{#1}}
\def\effectf#1{e_{\mathsf{#1}}}
\def\peffectf#1#2{e_{\mathsf{#2}}^{#1}}
\def\rankbyscore#1#2{r_{#1}^{#2}}
\def\ranksum{r_{\scoresum}}
\def\ranklex{r_{\mathsf{Lex}}}
\def\lex{\mathsf{Lex}}

\def\effmd{\effectf{md}}
\def\assign#1{\alpha[{#1}]}
\def\distmaxdispl{d_{\mathsf{md}}}
\def\distham{d_{\mathsf{Ham}}}
\def\distktau{d_{\mathsf{k}\tau}}

\def\nuknap{\nu_{\mathsf{knap}}}
\def\nucnf{\nu_{\mathsf{CNF}}}

\def\restrict#1#2{\left.#1\right|_{#2}}

\def\SHAPscore{\text{\normalfont\textsc{Shap}}}
\def\Shapleyvalue{\mathrm{Shapley}}

\DeclarePairedDelimiter{\size}{\lvert}{\rvert}
\DeclarePairedDelimiter{\set}{\lbrace}{\rbrace}
\def\M{\mathcal{M}}
\def\Mat{\mathbf{M}}
\def\S{\mathcal{S}}
\def\reals{\mathbb{R}}
\def\expectation{\mathbb{E}}
\def\prob{\mathbb{P}}
\def\supp{\operatorname{supp}}
\def\bigO{\mathcal{O}}

\newcommand*{\problemname}[1]{{\mathsf{#1}}}
\newcommand{\PREC}[1]{\problemname{PREC}\angs{#1}}
\newcommand{\EXP}[2]{\problemname{EXP}\angs{#1,#2}}
\newcommand{\PSHAP}[2]{\problemname{SHAP}\angs{#1,#2}}
\newcommand{\PShapley}[2]{\problemname{Shapley}\angs{#1,#2}}

\newcommand*{\scoresum}{\scoref{Sum}}
\newcommand*{\scoremax}{\scoref{Max}}
\newcommand*{\scoremin}{\scoref{Min}}

\newcommand*{\asc}{\mathsf{asc}}
\newcommand*{\desc}{\mathsf{dsc}}

\newcommand*{\complexityfont}[1]{\mathsf{#1}}
\newcommand*{\sharpP}{\ensuremath{\complexityfont{\#P}}}
\newcommand*{\FP}{\ensuremath{\complexityfont{FP}}}
\newcommand*{\NP}{\ensuremath{\complexityfont{NP}}}
\newcommand*{\RP}{\ensuremath{\complexityfont{RP}}}



\Copyright{Anonymous}

\acknowledgements{This work has been initiated in 
\href{https://www.dagstuhl.de/seminars/seminar-calendar/seminar-details/24032}{Dagstuhl Seminar 24032: \emph{Representation, Provenance, and Explanations in Database Theory and Logic}}.}

\begin{document}
\maketitle

\begin{abstract}
How important is the weight of a given column in determining the ranking of tuples in a table?
To address such an \emph{explanation question about a ranking function}, we investigate the computation of SHAP scores for column weights, adopting a recent framework by Grohe et al.\ [ICDT'24]. 
The exact definition of this score depends on three key components: 
(1) the \emph{ranking function} in use,
(2) an \emph{effect function} that quantifies the impact of using alternative weights on the ranking,
and (3) an underlying \emph{weight distribution}. We analyze the computational complexity of different instantiations of this framework for a range of fundamental ranking and effect functions, focusing on probabilistically independent finite distributions for individual columns.

For the ranking functions, we examine lexicographic orders and score-based orders defined by the summation, minimum, and maximum functions.
For the effect functions, we consider \emph{global}, \emph{top-$k$}, and \emph{local} perspectives: 
global measures quantify the divergence between the perturbed and original rankings,
top-$k$ measures inspect the change in the set of top-$k$ answers, and
local measures capture the impact on an individual tuple of interest.
Although all cases admit an additive fully polynomial-time randomized approximation scheme (FPRAS), we establish the complexity of exact computation, identifying which cases are solvable in polynomial time and which are \#P-hard. 
We further show that all complexity results, lower bounds and upper bounds, extend to a related task of computing the Shapley value of whole columns (regardless of their weight).

\end{abstract}

\section{Introduction}

A \e{ranking function} takes as input a set of records and produces a permutation over the set based on the entry values (commonly referred to as \e{features}) of the records. As ranking is abundant in decision making, it is natural to look for an \e{explanation} for the outcome of the ranking of a given dataset.
Explanations for ranking functions have received significant attention in the area of Information Retrieval (IR) over recent years~\cite{DBLP:conf/chiir/PenhaK022,DBLP:conf/ecir/LyuA23,DBLP:conf/fat/SinghA20,DBLP:conf/sigir/YuRA22} and before, as well described in surveys on ``explainable IR''~\cite{DBLP:journals/corr/abs-2212-07126,DBLP:journals/corr/abs-2211-02405}. 
Explanations for ranking can be largely categorized into two types: A \e{value-based} explanation aims to elucidate how the values of a certain record led to its being ranked in its position (or why it is included in, or excluded from, the top-$k$ answers); this is typically done by quantifying the contribution of every entry value to the status of the record~\cite{anahideh2022local,chowdhury2025rankshap,DBLP:conf/chiir/PenhaK022}. 
On the other hand, a \e{function-based} explanation aims to investigate what aspect in the ranking function has led to the outcome, from the point of view of a specific record or the entire permutation; this is typically done by quantifying the contribution of function components to the outcome~\cite{yang2018nutritional,gale2020explaining,
DBLP:journals/corr/abs-2403-16085,DBLP:conf/fat/SinghA20}. This work belongs to the latter kind.

Typically, the ranking function involves \e{parameters}. For example, if the ranking is determined by a score function that is a linear combination of the attributes of the record, then the parameters are the \emph{weights} attached to the features of the record. Previous work in the database community has investigated ways of selecting parameters so that the ranking function satisfies certain criteria~\cite{Abolfazl18stable,chen2023not}.
In this work, we study the contribution of the parameter choices to the outcome of the ranking. 
For that, we adopt the framework of Grohe et al.~\cite{DBLP:conf/icdt/GroheK0S24} for measuring the importance of parameters in the context of database queries, where the contribution of a parameter is determined by its SHAP score~\cite{lund17}. 
In contrast to previous work on the explanation of ranking functions, we focus on the computational complexity of calculating the contribution. 

\vspace{-0.5em}

\hide{
\begin{figure}[t]
    \centering
        \begin{tabular}{c | c c c c|} 
         \cline{2-5}
         id & $a_1$ & $a_2$ & $a_3$ & $a_4$ \\ 
         \cline{2-5}
         1 & 0 & $\ldots$ & 0 & $x$ \\ 
         2 & $b_1$ & $\ldots$ & $b_n$ & 0 \\
         3 & $b_1$ & $\ldots$ & $b_n$ & 0 \\     
        \cline{2-5}
        \end{tabular}
    \caption{\Cref{ex:intro}: ... \benny{??}
}
    \label{fig:intro}
\end{figure}
}

\hide{
\begin{example}
\label{ex:intro}
... Consider the $3 \times 4$ table in \cref{fig:intro}
and a weight $\mathbf{w} = (1,2,3,4)$.

... ranked by sum, highest to lowest

... Prior worked attempts to choose weights in order to get a desired outcome.
Example solution: $\mathbf{w'} = (4,3,2,1)$ brings the 3rd tuple to the top.

... Here we assume a distribution on the weights. Concretely, each of the weight components $w_i''$ is either the one from above $w_i$ or $0$ half-half.

\end{example}
 }

\subparagraph{SHAP scores and Shapley values.}
The SHAP score~\cite{lund17} is an instantiation of the Shapley value that, in turn, is used to attribute a share to each player in a cooperative game, where each coalition (set of players) gains some utility~\cite{Shapley}. The Shapley value, named after its inventor, is unique up to some axioms of rationality (e.g., the sum of shares adds up to the utility of the entire set)~\cite{Shapley,roth} and, besides applications in a plethora of domains,
has been used for explanations in data-centric fields such as Machine Learning~\cite{lund17,DBLP:conf/pkdd/DuvalM21,rozemberczki2022shapley}, 
IR~\cite{chowdhury2025rankshap,DBLP:journals/corr/abs-2403-16085}, and databases~\cite{DBLP:journals/pacmmod/0002OS24,DBLP:journals/sigmod/BertossiKLM23,DBLP:conf/sigmod/DeutchFKM22}. 
The SHAP score was originally proposed for the purpose of explaining machine-learned models, and particularly attributing responsibility for the outcome of the model to the feature values of a particular instance~\cite{lund17,DBLP:books/sp/KamathL21}. 

Specifically, SHAP is the Shapley value in the cooperative game where the feature values are the players, and the utility of a coalition $C$ is the expected effect (model's output) of the instance where the feature values of $C$ are used and the rest chosen randomly. In the framework of Grohe et al.~\cite{DBLP:conf/icdt/GroheK0S24}, the parameter values play the same role that feature values play in the application of SHAP to machine-learning 
explanations~\cite{lund17,DBLP:books/sp/KamathL21} and that feature values play in the application of SHAP to value-based explanations of ranking~\cite{pliatsika2025sharpnovelfeatureimportance,chowdhury2025rankshap}. Moreover, the Shapley value has recently been used in function-based explanations~\cite{DBLP:journals/corr/abs-2403-16085}.

Connecting to this work, consider a particular parameterized ranking function, along with a particular choice of values for its parameters. These define a ranking over the given a set of records. 
To apply the SHAP score, two additional components need to be determined. The first is a \e{distribution} from which parameter choices are assumed to be drawn. 
The second is an \e{effect function} that determines how different the ranking obtained from a specific (random) choice of parameters is from the base ranking with the original choice of parameters. 
The effect function captures different interpretations of what Heuss et al.~\cite{DBLP:journals/corr/abs-2403-16085} refer to as ``listwise feature attribution.'' 
Similarly to the study of Grohe et al.~\cite{DBLP:conf/icdt/GroheK0S24}, we restrict the discussion to simple distributions of parameters, namely finite and probabilistically independent parameters (i.e., ``fully factorized'' distributions~\cite{van_den_broeck2022tractability}).

From the computational perspective, the SHAP score entails an expectation (w.r.t.~a random choice of players) over an expectation (over the random parameter choices). Nevertheless, it was shown by Grohe et al.~\cite{DBLP:conf/icdt/GroheK0S24} that the SHAP score can also be viewed as a single expectation over an efficiently samplable space; therefore, the SHAP score of a parameter can be computed in polynomial time via sampling if we settle for a randomized additive approximation (FPRAS), under the mild assumption that the effect function is computable in polynomial time. This technique allows additive approximations also for more general classes of probability distributions, including fully factorized continuous distributions that allow efficient sampling. Moreover, if we make a \e{data complexity} assumption of a fixed number of parameters, then the SHAP score can be computed in polynomial time via an explicit enumeration of the probability space. We investigate the ability to compute the SHAP score exactly (and deterministically), for a given (non-fixed) number of parameters, in polynomial time. Van den Broeck et al.~\cite{van_den_broeck2022tractability} have established that there are polynomial-time reductions in both directions between the computation of the SHAP score and the computation of the \e{expected effect}, which is arguably a simpler notion. Hence, following their result, we focus mainly on the expected effect throughout the paper.

\subparagraph{Studied ranking and effect functions.}
Our study considers a variety of basic ranking functions and effect functions. For ranking, we consider the score functions of $\scoresum$, $\scoremin$, and $\scoremax$, and the lexicographic ordering (descending or ascending). In all of these functions, the parameters are coefficients (weights) over the record attributes. We view the dataset simply as a matrix (where every row corresponds to a record), and the parameters as per-column multiplicative weights. For the effects, we consider three types: \e{global} measures determine an effect over the entire ranking, \e{top-$k$} measures determine the impact on the set of top-$k$ answers, and \e{local} measures determine the impact on a single record. 
The global measures we consider are standard distances between permutations~\cite{DBLP:conf/aaai/Procaccia0Z15,DBLP:journals/tit/BargM10}, and specifically Kendall's tau, maximum displacement, 
and Hamming distances between the random ranking (due to the random choice of parameters) and the base ranking. The top-$k$ measures include the symmetric difference between the top-$k$ sets (for the base and random rankings) and the binary indicator of whether there is \e{any} difference between the top-$k$ sets.
The local measures include the change in the record's position and the change in its status of membership in the set of top-$k$ answers. A detailed example is presented at the end of \Cref{sec:problems}.

\vspace{-0.5em}

\subparagraph{Contributions.}
We begin with algorithms for a simpler task: given two records $r_1$ and $r_2$, compute the probability (for randomly chosen parameters) that $r_1$ precedes $r_2$ in the ranking. 
This problem is tractable for all the ranking functions we consider, while in the case of $\scoresum$ we make the necessary assumption that the numbers are given in unary representation (as the problem is \fpsharpp-hard\footnote{Recall that \fpsharpp is the class of functions computable in polynomial time using an oracle to some function in \#P. A function $F$ is \fpsharpp-hard if there is a polynomial-time Turing reduction from every function in \fpsharpp to $F$.
Such a problem is at least as hard as every problem in the polynomial hierarchy~\cite{DBLP:journals/siamcomp/TodaO92}.}
for the binary representation). 
Our results (summarized in \Cref{table:complexity}) show that, in general, this pairwise case suffices for the tractability of some of the effect functions, namely Kendall's tau and the tuple's position and top-$k$ membership of a row for small $k$; for the rest, we prove hardness results. We also consider the case where the number of records is bounded, that is, we have a small number of competitors. In this case, we prove that all problems become tractable, again with the exception of $\scoresum$ with binary numeric representation.

Finally, we consider another variation of a function-based explanation for ranking, now parameter-free. The goal is to compute the contribution of whole columns (attributes) to the ranking. More precisely, we consider the Shapley value of the cooperative game where the players are the columns, and the utility of a set of columns is the effect of the ranking obtained by considering only the columns in the set, while ignoring the rest. This task has been presented by Heuss et al.~\cite{DBLP:journals/corr/abs-2403-16085} where they refer to the action of ignoring a column as ``masking the feature vector,'' yet with no complexity analysis. We show that this problem reduces in polynomial time to the problem in the focus of this paper, namely computing the SHAP score of parameter choices. Hence, our algorithms for the SHAP score of parameters can be used for the Shapley value of columns. Moreover, we show that the hard cases of computing the SHAP scores are hard already for the Shapley value of columns.

\subparagraph{Organization.} The rest of the paper is organized as follows. We begin with preliminary definitions in the next section. 
In \Cref{sec:problems}, we describe the formal framework and computational problems that we study, namely the SHAP score of parameters, the expected effect, and the Shapley value of columns. We give algorithms for the expectation in \Cref{sec:algorithms} and establish lower bounds in \Cref{sec:hardness}. 
Finally, in \Cref{sec:shapley}, we discuss the extension of our results to the Shapley value of columns, and conclude in \Cref{sec:conclusions}. For space limitations, some of the proofs are given in the Appendix.

{
\begin{table}[t]%
\caption{Summary of the complexity results. The complexity for the Hamming distance is the same as MD (maximum displacement). The complexity results for the effects of the top-$k$ perspective are the same as the top-$k$ membership. All problems are solvable in polynomial time when the number of attributes is fixed. They are also solvable in polynomial time when the number of records is fixed, with the exception of $\scoresum$ with binary numerical representation. All \fpsharpp-hardness results are actually \fpsharpp-completeness, due to standard techniques probability computation in \fpsharpp~\cite{DBLP:conf/pods/GradelGH98,DBLP:journals/jacm/FaginKK11}.
}
\label{table:complexity}
    \centering
\def\tabhard{\cellcolor{orange!25}\fpsharpp-h.\xspace}
    \def\tabhardbg{\cellcolor{orange!25}}
    \def\tabP{\cellcolorblue P\xspace}
    \begin{tabular} 
    {|c|c|c|c|c|c|}
      \multicolumn{1}{c}{ }  & \multicolumn{2}{c|}{\e{Global}} 
        & \multicolumn{3}{c}{\e{Local}} \\
    \cline{1-6}
    \multicolumn{1}{|c|}{}& & & &
    \multicolumn{2}{c|}{Top-$k$ membership} \\ 
    \cline{5-6}
         \backslashbox{Ranking}{Effect} &  Kendall's $\tau$ & MD & Position & Fixed $k$ & Given $k$
         \\ \hline
        $\scoresum$ unary & \tabP & \multirow{5}{*}{\tabhard}  & \tabP & \tabhard  & \multirow{5}{*}{\tabhard}   
        \\ 
         $\scoresum$ binary & \tabhard  & \tabhardbg  & \tabhard  & \tabhard  & \tabhardbg  
         \\  $\scoremax$ $\asc$  / $\scoremin$ $\desc$ & \tabP & \tabhard   & \tabP  & \tabhard & \tabhard \\ $\scoremax$ $\desc$  / $\scoremin$ $\asc$ 
          & \tabP  &\tabhardbg & \tabP  & \tabP & \tabhardbg\\ 
        Lexicographic & \tabP  & \tabhardbg   & \tabP  & \tabhard & \tabhardbg \\\hline
    \end{tabular}
\end{table}
}

\section{Preliminaries}\label{sec:preliminaries}

We begin with preliminary concepts and terminology that we use throughout the paper.

\subparagraph{Matrices and permutations.}
For a natural number $n$, we denote by $[n]$ the set
$\set{1,\dots,n}$. We also denote by $\S_n$ the set of permutations over $[n]$, where a permutation in $\S_n$ is a bijective function $\pi:[n]\rightarrow[n]$. 
By $\piid$ we denote the \e{identity} permutation defined by $\piid(i)=i$ for all $i\in[n]$. By $\S$ we denote the set of all permutations, that is, 
$\S\defeq\cup_n{\S_n}$.
We denote by $\M^{n\times m}$ the set of all $n\times m$ matrices, over the rational numbers, with $n$ rows and $m$ columns, 
and by $\M$ the set of all matrices of all dimensions (i.e., $\M\defeq \cup_{n,m}\M^{n\times m}$).
If $\Mat\in\M^{n\times m}$ and $C\subseteq\set{1,\dots,m}$ represents a set of columns, then we denote by $\restrict{\Mat}{C}$ the $n\times|C|$ matrix obtained from $\Mat$ by removing all columns except those in $C$. \added{In other words, $\restrict{\Mat}{C}$ is the projection of $\Mat$ on $C$ under bag semantics.}
We may refer to a row of a matrix as a \e{tuple}. We denote such a row (and every numeric vector) $\vect v$ in boldface and its $j$th entry with $\vect v[j]$.

Computation-wise, we assume every number is represented as a pair $(a,b)$, standing for the rational number $a/b$, where $a$ and $b$ are integers in a binary representation. When we refer to a number as represented in \e{unary} representation, we mean an integer $a$ encoded as a string of length $a$.

\subparagraph{Ranking functions.}

\hide{
\christoph{Add to this paragraph that we can restrict our attention to
\begin{itemize}
    \item SUM ascending: $\ranksum$
    \item MAX ascending: $\rankbyscore{\scoremax}{\asc}$
    \item MAX descending: $\rankbyscore{\scoremax}{\desc}$
    \item LEX: in each column ascending: $\ranklex$
\end{itemize}}
}

By a \e{ranking function} we refer to a function $r:\M\rightarrow\S$ that maps every matrix $\Mat\in\M^{n\times m}$ to a permutation in $\S_n$; hence, $r$ ranks the rows of its input matrix $\Mat$.
Specifically, we will focus on several
ranking functions $r$:
\begin{itemize}
    \item Ranking by decreasing/increasing score of a row $(a_1,\dots,a_m)$ where, with the score being the sum of the $a_j$, denoted $\scoref{Sum}$, the maximum value among the $a_j$, denoted $\scoref{Max}$, or the minimum among the $a_j$, denoted $\scoref{Min}$. 
    For each scoring function $s$, we denote by $r_s^{\desc}$ and $r_s^{\asc}$ the rankings by decreasing and increasing $s$, respectively. For example, we have $r_{s}^{\desc}(i)<r_{s}^{\desc}(j)$, meaning that the $i$th row $\vect t_i$ precedes the $j$th row $\vect t_j$,  whenever $s(\vect t_i)>s(\vect t_j)$. 
    
    \item Ranking by the lexicographic ordering (left to right) over the rows, denoted $r_{\scoref{Lex}}$. 
\end{itemize}
In our analysis, the direction of the ranking with respect to the score (increasing/decreasing) is important for ranking by $\scoremin$ and $\scoremax$, as shown in \Cref{table:complexity}. This direction is not important for ranking by $\scoresum$ since we can switch between the directions by multiplying the entries by a negative number (or subtracting each value from the maximum value in the matrix); hence, we simply write $r_\scoresum$ as standing for $\rankbyscore{\scoresum}{\asc}$ (since $\rankbyscore{\scoresum}{\desc}$ is equivalent). In the lexicographic order, we assume for the same reason that the order of each column is ascending (i.e., lower numbers precede higher numbers).
Also for the same, we will not consider ranking by $\scoremin$ explicitly since it is the same as ranking by $\scoremax$ of the negated matrix in reversed direction.

For the framework to be well-defined, we need to handle 
tie-breaking. For that, we will use the ordering by the row indices; that is, if the rows $i$ and $i'$ are tied by the ordering and $i<i'$, then we give precedence to the $i$th row.

\subparagraph{Shapley value.} A \e{cooperative game} is a pair $(P,\nu)$ where $P$ is a finite set of players and $\nu:2^P\rightarrow\reals$ is a \e{utility function} that associates with every coalition $C\subseteq P$ a value $\nu(C)$, so that $\nu(\emptyset)=0$. The \e{Shapley value} of a player $p\in P$ is defined by the following formula~\cite{roth,Shapley}. 
\begin{equation}\label{eq:shapley}
    \Shapleyvalue(p, \nu) 
    = \sum_{C \subseteq P \setminus\set{p}} \frac{|C|!(|P|-|C|-1)!}{|P|!} \cdot 
    (\nu(C \cup \set{p}) - \nu(C))
    \end{equation}
Intuitively, we consider the situation where we select players iteratively without replacement, starting with the empty set; the Shapley value of a player $p\in P$ is the mean increase in utility when adding $p$. 

\subparagraph{SHAP score.}
Let 
$f:D^m\rightarrow\reals$ be an $m$-ary function over some domain $D$. Let $\Pi$ be a discrete distribution over $D^m$, and let $\vect w=(w_1,\dots,w_m)$ be a tuple from $\Pi$ with a nonzero probability. For $j=1,\dots,m$, the  \e{SHAP score} of $j$ with respect to (w.r.t.) $\vect w$, denoted $\SHAPscore(j,f,\vect w)$, is the value $\Shapleyvalue(j,\nu)$ for the cooperative game $(P,\nu)$ defined as follows~\cite{lund17}.
\begin{itemize}
    \item $P\defeq\set{1,\dots,m}$;
    \item $\nu(C)\defeq \expectation_{\vect u\sim\Pi}\left[f(\vect u)\mid u_\ell=w_\ell {\mbox{ for all } \ell\in C}\right]$.
\end{itemize}
That is, the utility of a subset $C$ of parameters is the expectation of $f$ over the distribution $\Pi$, conditioned on every parameter in $C$ having its specific value from $\vect w$.

\section{Framework and  Computational Problems}
\label{sec:problems}

Our goal is to compute the contribution of columns and column weights to the ranking of the rows of a given matrix. For that, we need to reason about how different the ranking would be had we eliminated certain columns or changed their weights. Therefore, we need to adopt a measure of difference between permutations. We refer to such measures as \e{effect functions}, as they determine the effect of the column alteration. 

\subparagraph{Effect functions.}
We are given a matrix $\Mat$ 
and a ranking function $r$ with \emph{original order} $\pi_0 \defeq r(\Mat)$.
As a result of applying the ranking function with different weights on the columns (or only on a subset of columns), we get a different permutation $\pi$.
We will focus on three classes of effect functions:
\begin{itemize}
\item Global perspective: \e{How far is $\pi$ from the original permutation?} For that, we can use several notions of distances between permutations. We give here three conventional distances, where the first two were used for ranking explanation in the long version of the RankSHAP work~\cite{chowdhury2025rankshap}.
\begin{itemize}
    \item Kendall's tau: $\effectf{k\tau}(\pi)
\defeq \sum_{1\leq i<j\leq n}\one_{\pi_0(i)<\pi_0(j)\land \pi(i)>\pi(j)}$ determines the number of pairwise disagreements (swaps) between $\pi$ and the original order. This effect function has the range $[\frac{n(n-1)}{2}]$.

\item Maximum Displacement (MD) distance: $\effectf{md}(\pi)
\defeq \max_{1\leq i\leq n}\size{\pi(i)-\pi_0(i)}$ is the maximum difference of an item's position between the two permutations. This difference is a number in $[n-1]$.
\item Hamming distance: $\effectf{Ham}(\pi)
\defeq \sum_{1\leq i\leq n}\one_{\pi(i)\neq\pi_0(i)}$ determines the number of positions where the tuple is different between the original and permuted orders. It can attain values in $[n]$.
\end{itemize}
\item Top-$k$ perspective: \e{What is the impact on the set of top-$k$ answers?} For that, let $T=\set{\pi^{-1}(1),\dots,\pi^{-1}(k)}$ and 
$T_0=\set{\pi^{-1}_0(1),\dots,\pi^{-1}_0(k)}$
we will consider two functions. 
\begin{itemize}
   \item Top-$k$ difference: $\peffectf{k}{\mathrm{\Delta}}(\pi)
\defeq \size{T\cup T_0}-\size{T\cap T_0}$ 
is the size of the symmetric difference between the sets of top-$k$ elements. 
\item Top-$k$ any-change: $\peffectf{k}{any}(\pi)
\defeq \one_{T\neq T_0}$ determines whether there is \e{any} change in the top-$k$ elements. 
\end{itemize}
\item Local perspective: \e{What is the impact on a specific row $i$?} For that, we will consider two functions. 
\begin{itemize}
\item Position: $\effectf{pos}(i,\pi)
\defeq \pi(i)-\pi_0(i)$ is the change of position of the $i$th row and is in the range $\set{-n+1, \ldots, n-1}$.
\item Top-$k$ membership: $\peffectf{k}{top}(i,\pi)
\defeq 
\one_{\pi(i)\leq k} - \one_{\pi_0(i)\leq k}$ determines how the $i$th row changes its membership in the top-$k$ tuples. Its range is either $\set{-1,0}$ or $\set{0 ,1}$.
\end{itemize}
\end{itemize}
In the analysis we conduct in the remainder of this paper, we focus on four effects: Kendall's tau, MD, position, and top-$k$ membership. 
The complexity results for the Hamming distance are the same as those of MD (see \cref{sec:appendix:hardness-hamming}). 
The complexity results for the impacts of the top-$k$ perspective are the same as the top-$k$ membership (see \cref{sec:appendix:top-k}).

\newcommand{\problem}[4]{
\begin{center}
\begin{tabular}{|rl|}\hline
\textbf{Problem:} & #1 \\
\textbf{Fixed:} & #2 \\
\textbf{Input:} & #3 \\
\textbf{Goal:} & #4\\
\hline
\end{tabular}
\end{center}
}

\subparagraph{Computing SHAP scores.}
Next, we assume that the columns of the input matrix $\Mat$ are weighted, and the goal is to compute the SHAP score of a given weight. To make it precise, we assume that our data is a pair $(\Mat,\vect w)$, where $\Mat\in\M^{n\times m}$ and $\vect w=(w_1,\dots,w_m)\in\mathbb{Q}^m$ is a sequence of column weights. The pair $(\Mat,\vect w)$ represents the pair $\Mat\circ \vect w$, which is the matrix $\Mat'$ obtained by multiplying the $j$th column by $w_i$ for $i=1,\dots,m$; that is, we apply element-wise multiplication by $\vect w$ to each row of $\Mat$, or, in other words, $\Mat' = \mathsf{diag}(\vect w) \cdot \Mat$
where $\mathsf{diag}(\vect{w})$ is the diagonal matrix whose diagonal elements are given by $\vect{w}$.

We will also assume that, in addition to $\Mat$ and $\vect w$, we are given a finite 
probability distribution $\Pi$ over $\mathbb Q^m$. 
We restrict our complexity study to column-wise independent distributions (i.e., \e{fully factorized} distributions~\cite{van_den_broeck2022tractability}),
hence, $\Pi$ is represented by finite distributions $\Pi_1,\dots,\Pi_m$ over $\mathbb Q$ for each weight parameter, each distribution given as a list of value-probability-pairs; 
the total number of such pairs is denoted by $|\Pi|$.

Let $r$ be a ranking function and $e$ an effect function. Given the pair $(\Mat,\vect w)$ and a column number $j$, our goal is to compute $\SHAPscore(j,f,\vect w)$, where $f(\vect u)\defeq -e(r(\Mat\circ \vect u))$, that is, the inverse of the effect on the ranking of replacing the weights of $\vect w$ with those of $\vect u$. 
We denote this value by 
$\SHAPscore\angs{r,e}(j,\Mat,\vect w,\Pi)$.
The reason for the inverse (using minus) is that the closer we are to the original ranking order $\pi_0$ (i.e., the less effect), the higher we deem the contribution to the the actual $\pi_0$ in place.

\problem
{$\PSHAP{r}{e}$: SHAP score computation}
{Ranking function $r$ and effect function $e$}
{Matrix $\Mat$, weight vector $\vect w$, distribution $\Pi$, and column number $j$ of $\Mat$}
{Compute $\SHAPscore\angs{r,e}(j,\Mat,\vect w,\Pi)$}

In this and the other problems we consider, when the effect function $e$ is local, then the input also includes the row number $i$ of the effect. In addition, when $e$ is the top-$k$ membership, then the input also includes $k$ (unless we explicitly state that $k$ is assumed to be fixed).

It has been established by Van den Broeck et al.~\cite{van_den_broeck2022tractability} that, for finite and probabilistically independent parameters, calculating the SHAP score is computationally equivalent to calculating the expected value of $f$, which is arguably simpler to handle.

\begin{theorem}\label{thm:shap-expectation}
{\e{\!\!\!\protect{\cite[Theorem~2]{van_den_broeck2022tractability}}}}\,
Let $f$ be a function that can take any number of numerical arguments and can be computed in polynomial time. The following problems are polynomially (Turing) reducible to each other.
\begin{enumerate}
    \item Compute $\SHAPscore(j,f,\vect w)$, given $\vect w$, $j$ and $\Pi$.
    \item Compute $\expectation_{\vect u\sim\Pi}[f(\vect u)]$, given $m$ and $\Pi$.
\end{enumerate}
\end{theorem}
Hence, we will also study the following problem.

\problem
{$\EXP{r}{e}$: Expectation computation}
{Ranking function $r$ and effect function $e$}
{Matrix $\Mat$ and distribution $\Pi$}
{Compute $\expectation_{\vect u\sim\Pi}[e(r(\Mat\circ \vect u))]$}

\subparagraph{Computing Shapley values.}
Let $r$ be a ranking function and $e$ an effect function. Given a matrix $\Mat\in\M^{n\times m}$ and a column number $j\in\set{1,\dots,m}$, the Shapley value of the $j$th column is defined as $\Shapleyvalue(i,\nu)$ for the game $(P,\nu)$ defined as follows:
\begin{itemize}
\item $P\defeq\set{1,\dots,m}$.
\item $\nu(C)\defeq -e(r(\restrict{\Mat}{C}))$, that is, the inverse of the effect on the ranking we obtain by applying $r$ to only the columns of $C$ (with the same rationale for inverse as SHAP).\footnote{The ranking function $r$ applied to the empty set of columns always yields the matrix order $\pi_{\mathsf{id}}$ since all tuples of $\restrict{\Mat}{\emptyset}$ are identical. In particular, this definition does not depend on the value of the scoring function on the empty set.}
\end{itemize}
We denote this value by 
$\Shapleyvalue\angs{r,e}(j,\Mat)$.

\problem
{$\PShapley{r}{e}$: Shapley value computation}
{Ranking function $r$ and effect function $e$}
{Matrix $\Mat$ and a column number $j$ of $\Mat$}
{Compute $\Shapleyvalue\angs{r,e}(j,\Mat)$}

\subparagraph{Approximation.}
Before moving to exact algorithms, we want to emphasize that each of the three problems $\EXP{r}{e}$, $\PSHAP{r}{e}$, and $\PShapley{r}{e}$ can be efficiently additively approximated  whenever $r$ and $e$ can be computed in polynomial time and the range of $e$ is polynomial in the number of rows. This is the case for all ranking and effect functions introduced in this section. As we will show in \cref{sec:hardness:top-k}, this is not always the case for multiplicative approximations.
The additive approximation can be done via basic Monte-Carlo-sampling since all three problems can be expressed as the expectation of a random variable over a probability space that allows efficient sampling: $\Pi$ for $\EXP{r}{e}$, $\S_m$ for $\PShapley{r}{e}$, and $\Pi \times \S_m$ for $\PSHAP{r}{e}$. We can use this insight for efficient additive approximations also in the case of not necessarily independent discrete parameter distributions under very mild assumptions (see \cite{DBLP:conf/icdt/GroheK0S24}) or independent continuous parameter distributions that can be sampled efficiently.\footnote{For continuous distributions, the definition of the SHAP score via conditional representations is not well-defined since the set we condition on has measure $0$. In the case of independent distributions, we can derive a compatible definition via disintegration, that is, fixing the parameters in $C$ to their reference values and choosing the others at random.}
This simple idea was further improved to allow for faster additive approximations (e.g. \cite{CASTRO2017180}). Hence, our framework could be applied in practice, even in cases where we show that exact computation is hard. 

\smallskip

\subparagraph{A detailed example.}

Before we start our complexity analysis, let us first look at a small example that demonstrates the framework in detail. Consider the following matrix $\Mat$:

 \begin{center}
    \begin{tabular}{c| c c |} 
     \cline{2-3}
     id & $a_1$ &  $a_2$ \\ 
     \cline{2-3}
     1 & 20 & 26 \\ 
     2 & 30 & 13 \\
     3 & 40 & 0 \\
     4 & 0 & 39 \\
    \cline{2-3}
    \end{tabular}
\end{center}

Assume that we rank the row by $\scoresum$ in descending order with a scoring parameter $\vect u \in \set{1,2}\times \set{1,2}$. We want to measure the impact of choosing the weight vector $\vect w = (1,1)$ among the possible choices for $\vect u$ (where we assume a uniform distribution).

To measure this, we first determine the rankings for each $\vect u$: 
For $\vect u = \vect w = (1,1)$, the scores are $(46, 43, 40, 39)$, so we obtain the ranking $(1,2,3,4)$. In the same way, we obtain $(4,1,2,3)$ 
for $\vect u = (1,2)$, $(3,2,1,4)$ for $\vect u = (1,2)$, and again $(1,2,3,4)$ for $\vect u = (2,2)$.

 \begin{center}
    \begin{tabular}{r| c | r | r | r } 
     $\vect u\,\,\,$    &scores         & permutation $\pi$    &  $\effectf{k\tau} = \distktau(\pi, \pi_0)$ & $\effectf{pos}(4, \pi)$\\ 
     \hline
     $\vect w = (1,1)$  & (46,43,40,39)  & $\pi_0=(1,2,3,4)$ & 0    &4 \\ 
     $(1,2)$            & (72,56,40,78)  & $(4,1,2,3)$       & 3    &1 \\ 
     $(2,1)$            & (66,73,80,39)  & $(3,2,1,4)$       & 3    &4 \\ 
     $(2,2)$            & (92,86,80,78)  & $(1,2,3,4)$       & 0    &4 \\ 
    \end{tabular}
\end{center}

Now, we need to choose an effect function $e$ that measures the distance between two rankings. 
We will look at the Kendall's tau distance $\effectf{k\tau}$ and the position of the fourth row $\effectf{pos}(4, \pi)$. We observe that $\distktau((1,2,3,4), (4,1,2,3)) = \distktau((1,2,3,4), (3,2,1,4)) = 3$ and that the row $4$ is in the last position of all the rankings except for $(4,1,2,3)$ where it is first.

For the SHAP score calculation, the valuation of a set of columns $C$ is the expected value of the negated effect function when we fix the parameter value for the columns in $C$ to the value of $\vect w$ and choose the rest at random. For example, for $C = \set{2}$, we obtain for $e = \effectf{k\tau}$
\begin{align*}
\nu_{\mathsf{k}\tau}(\set{2}) &= -\frac{\effectf{k\tau}\big(\ranksum(\Mat \circ (1,1))\big) + \effectf{k\tau}\big(\ranksum(\Mat \circ (2,1))\big)}{2}
\\
&= -\frac{\distktau\big((1,2,3,4), (1,2,3,4)\big) + \distktau\big((1,2,3,4), (3,2,1,4)\big)}{2} = -\frac{3}{2},
\end{align*}
and for $e = \effectf{pos}(4, \pi)$, we obtain
\[
\nu_{\mathsf{pos}(4)}(\set{2}) = -\frac{(4-4) + (4-4)}{2} = 0.
\]
Calculating the valuation $\nu$ of every $C \subseteq \set{1,2}$ gives the following. 
For $e = \effectf{k\tau}$, we obtain 
\begin{align*}
\SHAPscore(1) = & \frac{1}{2}\big(\nu_{\mathsf{k}\tau}(\set{1,2}) - \nu_{\mathsf{k}\tau}(\set{2})\big) + \frac{1}{2}\big(\nu_{\mathsf{k}\tau}(\set{2}) - \nu_{\mathsf{k}\tau}(\emptyset)\big) 
\\
= & \frac{1}{2}\big(0 - (-\frac{3}{2})\big) + \frac{1}{2}\big(-\frac{3}{2} - (-\frac{3}{2})\big) =  \frac{3}{4}
\end{align*}
and $\SHAPscore(2) = \frac{3}{4}$ as well, so both parameter choices have the same impact w.r.t. global changes in the ranking measured by $\distktau$. This reflects the symmetry in the 4th column of table above. In contrast, for $e = \effectf{pos}(4, \pi)$, we obtain $\SHAPscore(1) = \frac{3}{8}$ and $\SHAPscore(2) = -\frac{9}{8}$. This matches our intuition that fixing the second weight to $1$ is bad for the ranking of row $4$, since the only good weighting vector for row $4$ is $\vect u = (1,2)$.

In the remainder of the paper, we study the complexity of exactly solving the computational problems we defined in this section.

\section{Exact Algorithms for Expectation (and SHAP Scores)}
\label{sec:algorithms}

In this section, we show the tractability results of \Cref{table:complexity} for the problems $\EXP{r}{e}$ of computing the expected effect. Combined with \Cref{thm:shap-expectation}, these give the tractability results of the table for the SHAP score. We begin a more basic problem that we later use for expectation computation.

\subsection{Algorithms for the Pairwise Case}\label{sec:pairwise}
We first give algorithms for the following problem: \e{What is the probability that a given tuple $x$ precedes another tuple $y$ in the ranking?}
Since we only care about their relative order, we can ignore all the other tuples in the table and assume $n=2$.
Although this pairwise case may seem highly restricted, we will see in the next section that it forms the foundation of all our algorithms.

\problem
{$\PREC{r}$: Precedence probability computation}
{Ranking function $r$}
{Matrix $\Mat$ with two rows, distribution $\Pi$}
{Compute $\prob_{\vect u\sim\Pi}[s(1)<s(2)]$ where $s$ stands for $r(\Mat\circ \vect u)$}

A naive way to compute the probability of precedence is to iterate through the exponentially many choices of weights $\vect u$ and check whether tuple precedence holds in the resulting order.
Our more efficient algorithms exploit the fact that many of these choices can be grouped together, allowing us to consider only a polynomial number of possibilities.
This can be achieved for all the ranking functions discussed in this paper,
except for $\scoref{Sum}$, where our algorithm is pseudo-polynomial,
that is, the number of cases it considers is polynomial in the magnitude of the values we are given.

\begin{theorem}
\label{thm:prec_algs}
The problems
$\PREC{\rankbyscore{\scoremax}{\asc}}$, $\PREC{\rankbyscore{\scoremax}{\desc}}$,
and $\PREC{\ranklex}$ 
are solvable in polynomial time.
The problem $\PREC{\ranksum}$ is solvable in polynomial time if the matrix $\Mat$ and the weights in $\Pi$ are integers encoded in unary.
\end{theorem}

In the following, we give the algorithm for each ranking function in \Cref{thm:prec_algs}.
For the sake of readability, we use $\vect x$ and $\vect y$ for the two tuples $\vect{t}_1$ and $\vect{t}_2$ in the matrix.
Our goal is to compute the probability that $\vect x$ precedes $\vect y$, and without loss of generality, we assume that this does not hold in the event of a tie.

\subparagraph{Sum ranking.}
For the $\scoresum$ ranking function, our algorithm relies on dynamic programming and follows a similar approach to the standard pseudo-polynomial algorithm for the counting knapsack problem~\cite{gopalan11knapsack}.
Specifically, we consider a subproblem where the matrix $\Mat$ is restricted to columns $j, \ldots, m$ and compute the probability
that tuple $\vect y$ surpasses tuple $\vect x$ by (strictly) more than $s$ in total score.
This can be expressed recursively by considering
how the score difference changes when removing the $j$-th column
for each assignment to the weight $u_j$:

\begin{align}
&R[j, s] = \sum_{v_j \in \supp(\Pi_j)} \prob[u_j = v_j] \cdot R[j + 1, s - v_j (\vect{y}[j] - \vect{x}[j])]\label{eq:dp_sum}\\
&R[m+1, s] = 1  \text{ if } s > 0, \mbox{ and }
R[m+1, s] = 0  \text{ if } s \leq 0
\end{align}

The final answer is $R[1, 0]$.
The running time is $\bigO(|\Pi| \cdot ||s||)$ where $||s||$ is the number of possible values for $s$ in our dynamic program.
If the input consists of integers given in unary, then $||s||$ is polynomial in the input size.
Indeed, suppose the matrix $\Mat$ contains integers in the range $[-c, c]$ and the weights in $\Pi$ are in the range $[-b, b]$.
Then $||s||$ will be bounded by $\sum_{j=1}^m b \cdot 2c$, which is in
$\bigO(mbc)$.

\subparagraph{Max ranking.}
For the $\scoremax$ ranking function, the key idea is that the ``winner'' in the comparison between the two tuples is determined by a single column---the one with the maximum value.
Thus, we can break the problem into distinct cases, each corresponding to a specific column determining the outcome.

\begin{example}
\label{ex:max_pairwise}
Let $\vect x=(3, 5, 2)$ and $\vect y=(4, 1, 6)$ be two tuples, and suppose that the weights $\vect u$ are drawn uniformly from $\set{0, 1}$, and the ranking function is $\rankbyscore{\scoremax}{\asc}$.
For $\vect x$ to be ranked before $\vect y$, we need $\max(u_1 \cdot 3, u_2 \cdot 5, u_3 \cdot 2) < \max(u_1 \cdot 4, u_2 \cdot 1, u_3 \cdot 6)$.
We consider three distinct cases where this event occurs:
(Case 1) The first $\vect y$-coordinate $u_1 \cdot 3$ is the overall maximum. 
We iterate over possible values of $u_1$, and for each one, we determine the valid weights for $u_2$ and $u_3$ individually.
If $u_1 = 0$, then $u_1 \cdot 3$ cannot be the winner.
If $u_1 = 1$, then the valid sets of weights are:
$u_1 \in \{1\}$, $u_2 \in \{0\}$, $u_3 \in \{0\}$. This occurs with probability $1/8$.
(Case 2) The second $\vect y$-coordinate $u_2 \cdot 1$ is the overall maximum. This cannot happen.
(Case 3) The third $\vect y$-coordinate $u_3 \cdot 6$ is the overall maximum. Here, we obtain the sets $u_1 \in \{0,1\}$, $u_2 \in \{0,1\}$, $u_3 \in \{1\}$. All combinations of these weights are valid, thus the probability is $4/8$.
Summing up the probabilities from the valid cases, we obtain the final result $5/8$. 
\qed
\end{example}

\Cref{alg:pairwise_minmax} builds on the logic of our example,
decomposing the comparison of two maximum predicates into $m$ distinct cases,
and computing their probabilities efficiently.
Some care is required in the event of ties among columns of $\vect y$; in such cases, we elect the first occurring maximum to be the winner, ensuring that the $m$ cases remain distinct.
The decomposition can be expressed as follows:

\begin{align}
    & (\max(\vect y \circ \vect u) > \max(\vect x \circ \vect u)) \equiv \bigvee_{j \in [m]} \bigvee_{v_j \in \supp(\Pi_j)} W_{j,v_j} \\
    & W_{j,v_j} = (v_j \cdot \vect y[j] > v_j \cdot \vect x[j]) \wedge
    (\bigwedge_{k \notin [j]} v_j \cdot \vect y[j] > u_k \cdot \vect x[k]) \wedge \\
    \nonumber
    & \qquad \qquad (\bigwedge_{k \in [j-1]} v_j \cdot \vect y[j] > u_k \cdot \vect y[k]) \wedge
    (\bigwedge_{k \in [j+1,m]} v_j \cdot \vect y[j] \geq u_k \cdot \vect y[k])
\end{align}

The algorithm constructs these distinct events $W_{j,v_j}$ and for each one, it determines sets of valid weights for $\vect u$, except for position $j$ where the weight has been fixed to $v_j$.
Finally, the probability of these events is computed efficiently, since the distributions $\Pi_j$ are independent. 
The running time of the algorithm is $\bigO(|\Pi|^2)$.

\begin{algorithm}[t]
\textbf{Input}: tuples $\vect x$ and $\vect y$, distribution $\Pi$ \algocomment{Assuming $\vect y$ is before $\vect x$ in the tie-breaking scheme}\\
\textbf{Output}: probability that $\vect x$ appears before $\vect y$ under $\rankbyscore{\scoremax}{\asc}$ ranking\\

\def\res{\mathrm{result}}
$\res \defeq 0$\;

\For{column $j$ in $[m]$}
{
    \For{weight $v_j \in \supp(\Pi_j)$}
    {
        \algocomment{Find cases where $v_j \cdot \vect y[j]$ is the winner}\;
        $U \defeq \{\emptyset, \ldots, \emptyset \}$ \algocomment{Data structure with one list of satisfying weights per column}\;
        \For{column $k$ in $[m]$}
        {
            \algocomment{Build the list of weights for column $k$}\;
            \lIf{$k = j$ \textbf{and} $v_j \cdot \vect y[j]  > v_j \cdot \vect x[j]$}{
                $U[j] \defeq \{ v_j\}$
            }

            \If{$k \neq j$}{

                $U_1 \defeq \{v_k \in \supp(\Pi_k) \:|\: v_j \cdot \vect y[j]  > v_k \cdot \vect x[k] \}$ \algocomment{Strictly beat $\vect x$ on column $k$}\;
                
                \If{$k < j$}{
                $U_2 \defeq \{v_k \in \supp(\Pi_k)  \:|\: v_j \cdot \vect y[j]  > v_k \cdot \vect y[k] \}$ \algocomment{Strictly beat $\vect y$ on column $k$}\;
                }
                \Else{
                $U_2 \defeq \{v_k \in \supp(\Pi_k) \:|\: v_j \cdot \vect y[j]  \geq v_k \cdot \vect y[k] \}$ \algocomment{Beat $y$ on column $k$}\;
                }
                $U[j] \defeq U_1 \cap U_2$\;
            }
        }
        $\res \pluseq \prod_{k \in [m]}  \sum_{v_k \in U[k]} \prob[u_k = v_k] $ \label{alg_line:sum_product}\;
    }
}
\Return{$\res$}\;

\caption{Precedence computation for ranking by $\scoremax$ ascending.}
\label{alg:pairwise_minmax}
\end{algorithm}

\subparagraph{Lexicographic ranking.}
We now move on to ranking by lexicographic orders.
Our algorithm again follows a decomposition into non-overlapping events,
each corresponding to a specific column and its weight determining the outcome.
For this to happen, all preceding columns need to be tied between the two tuples.
The decomposition works as follows:

\begin{align}
    & ((\vect y \circ \vect u) >_{\mathsf{Lex}} (\vect x \circ \vect u) \equiv \bigvee_{j \in [m]} \bigvee_{v_j \in \supp(\Pi_j)} W_{j,v_j}'\\
    & W_{j,v_j}' = (v_j \cdot \vect y[j] > v_j \cdot \vect x[j]) \wedge
    (\bigwedge_{k \in [j-1]} v_j \cdot \vect y[j] = u_k \cdot \vect x[k])
\end{align}

Similarly to the case of $\scoremax$, we construct these events,
determine the valid weight assignments $\vect u$ for each one, 
and calculate their probability.
The running time is $\bigO(|\Pi|^2)$.

\subparagraph{Recovering the satisfying weights.}
In some cases, we need not only the probability that tuple $\vect x$ is before tuple $\vect y$, but also a succinct representation of the weights for which this is true.
Fortunately, our decomposition-based algorithms for $\scoremax$ and $\lex$ provide this directly.
We simply need to replace summation with union and multiplication with set product in \Cref{alg_line:sum_product} of \Cref{alg:pairwise_minmax}.

\begin{observation}
\label{obs:list_of_weights}
The algorithms for $\PREC{\rankbyscore{\scoremax}{\asc}}$, $\PREC{\rankbyscore{\scoremax}{\desc}}$,
and $\PREC{\ranklex}$ can additionally return the weights $\vect u$ for which precedence holds as a disjoint list,
where each element consists of a set of valid weights per column, i.e., $\bigtimes_{j \in m} \bigcup (u_j=v_j)$.
\end{observation}

\subsection{From Pairwise Impact to Ranking Impact}

Having established efficient algorithms for the pairwise case, we now extend them to compute the expected effect on the overall ranking.
When the effect function measures the position of a tuple or the Kendall's tau distance of the entire permutation, we can reduce the problem to a polynomial number of calls to $\PREC{r}$.
For both measures, we can write the expected effect as the expectation of a sum of indicator variables for tuple precedence and then use the linearity of expectation.

\begin{restatable}{theorem}{thmmaxlexktauposptime}
\label{thm:max-lex-ktau-pos-ptime}
For ranking function $r \in \set{\rankbyscore{\scoremax}{\asc}, \rankbyscore{\scoremax}{\desc},
\ranklex}$
and effect function $e \in \set{\effectf{k\tau}, \effectf{pos}}$,
the problem $\EXP{r}{e}$ is solvable in polynomial time.
For ranking function $r=\ranksum$
and effect function $e \in \set{\effectf{k\tau}, \effectf{pos}}$,
the problem $\EXP{r}{e}$ is solvable in polynomial time if the matrix $\Mat$ and the weights in $\Pi$ are integers encoded in unary.
\end{restatable}

The final case that admits a polynomial-time algorithm is the expected top-$k$ membership for fixed $k$ and ranking by $\rankbyscore{\scoremax}{\desc}$.
Interestingly, our approach only applies to a descending order,
and in fact, we prove in the next section that the problem is hard
when the order is ascending.
The key idea of the algorithm is that, for this specific ranking function, we can compute the probability that the tuple of interest beats a small subset of other tuples.

\begin{figure}
\begin{center}
\begin{tabular}[b]{c c c} %
    \begin{tabular}{c| c c c|} 
        \cline{2-4}
        $\vect{t}_1$ & 3 & 1 & 0 \\ 
        $\vect{t}_2$ & 1 & 0 & 4 \\
        $\vect{t}_3$ & 0 & 1 & 1 \\
        $\vect{t}_4$ & \cellcolorblue 2 & \cellcolorblue 2 & \cellcolorblue 2 \\
        \cline{2-4}
    \end{tabular}
    &
    \begin{tikzpicture}[baseline=(current bounding box.south)]
        \fill[gray] (0,-0.15) rectangle (1.5,0.15);
        \fill[gray] (1.5,-0.3) -- (2,0) -- (1.5,0.3) -- cycle;
        \node[above,yshift=5pt] at (0.92,0) {\small Is $\vect{t}_4$ the top-$1$?};
    \end{tikzpicture}
    &
    \begin{tabular}{c| c c c|} 
        \cline{2-4}
        $\vect{t'}$ & 3 & 1 & 4 \\ 
        $\vect{t}_4$ & \cellcolorblue 2 & \cellcolorblue 2 & \cellcolorblue 2 \\
        \cline{2-4}
    \end{tabular}
\end{tabular}
\end{center}
\caption{\Cref{ex:max_topk}: The transformation to the pairwise case for $\rankbyscore{\scoremax}{\desc}$ and top-$k$ membership.}
\label{fig:max_topk}
\end{figure}

\begin{example}
\label{ex:max_topk}
Suppose the matrix $\Mat$ contains tuples $\vect{t}_1, \vect{t}_2, \vect{t}_3, \vect{t}_4$ as shown in \Cref{fig:max_topk}, the ranking function is $\rankbyscore{\scoremax}{\desc}$,
and our goal is to compute the probability that $\vect{t}_4$ is ranked top-1.
We can merge the competitor tuples $\vect{t}_1, \vect{t}_2, \vect{t}_3$ into a single competitor $\vect{t'}$ that retains the largest value per column.
Now, $\vect{t}_4$ is top-1 precisely when it precedes $\vect{t'}$ in the ranking.
This merging has effectively reduced the problem to the pairwise problem $\PREC{\rankbyscore{\scoremax}{\desc}}$, which we have already established can be solved efficiently.\qed
\end{example}

The example illustrates how the algorithm works when $k=1$. For $k > 1$,
we apply the principle of inclusion-exclusion to express the event that our tuple is in the top-$k$ as the intersection
of events where it ranks above specific subsets of other tuples.

\begin{restatable}{theorem}{thmmaxtopkptime}
\label{thm:max-topk-ptime}
The problem $\EXP{\rankbyscore{\scoremax}{\desc}}{\peffectf{k}{top}}$ with $k$ as a fixed parameter is solvable in polynomial time.
\end{restatable}

Later, we will show that the assumption of a fixed $k$ is necessary since the problem becomes intractable otherwise (\Cref{thm:top-k-hard}).

\subsection{Bounding the Matrix Dimensions}
\label{sec:bounded_dim}

We conclude this section by investigating the consequences of bounding one of the two dimensions of our matrix by a fixed constant.
This allows us to obtain a clearer picture of the parameters that make the problem hard.
We find that such a restriction makes the problem significantly easier and, in most cases (with the exception of $\scoresum$ with arbitrary values), both dimensions need to be non-fixed for the problem to be hard.

\begin{restatable}{theorem}{thmboundeddim}
\label{thm:bounded_dim}
The following hold.
\begin{enumerate}
\item If the number of columns $m$ of the matrix $\Mat$ is bounded by a constant, then the problem $\EXP{r}{e}$ is solvable in polynomial time for any ranking function $r$ and effect function $e$ computable in polynomial time.
\item If the number of rows $n$ of the matrix $\Mat$ is bounded by a constant, then the problem $\EXP{r}{e}$ is solvable in polynomial time for any effect function $e$ computable in polynomial time and
\begin{enumerate}
\item the ranking function $r$ is in $\set{\rankbyscore{\scoremax}{\asc}, \rankbyscore{\scoremax}{\desc},\ranklex}$, or
\item $r$ is the ranking function $\ranksum$ and, additionally, the matrix $\Mat$ and weights in $\Pi$ are integers encoded in unary.
\end{enumerate}
\end{enumerate}
\end{restatable}

We note that the proof is straightforward for the cases where $m$ (the number of columns) is fixed, since we can then materialize the entire probability space in an explicit representation. The proof of tractability for a fixed number of $n$ of rows is more involved.

\section{Intractable Cases for Expectation (and SHAP Scores)}\label{sec:hardness}

In this section, we show that the tractability results in \cref{sec:algorithms} are complete in the sense that expectation computation (and SHAP score computation) are $\FP^{\sharpP}$-hard for all remaining combinations of $r$ and $e$.

\subsection{Ranking based on Summation}\label{sec:hardness:sum}
We start with the observation that precedence probability is hard for ranking by $\scoref{Sum}$.

\begin{restatable}{theorem}{thmsumprechard}\label{thm:sum-prec-hard}
     The problem $\PREC{\ranksum}$ is $\FP^{\sharpP}$-hard, if the input matrix $\Mat$ is encoded in binary.
\end{restatable}

The proof of this theorem is given in \cref{sec:appendix:miss-hardness}. We use the following lemma to prove this claim by reducing from the counting knapsack problem which, given natural numbers $b_1, \ldots, b_\ell$ and a number $d$ all encoded in binary, asks for the number of subsets $S \subseteq [\ell]$ with $\sum_{i \in S} b_i \leq d$. This problem is known to be $\sharpP$-hard (see, e.g.,~\cite{DBLP:journals/cpc/DyerFKKPV93,gopalan11knapsack}).

\begin{lemma}\label{lem:sum-hard-matrix}
    Let $b_1, \ldots, b_\ell$ and $d$ be an input to the counting knapsack problem. Then, there is a matrix $\Mat \in \M^{2\times (\ell+1)}$ with the following property: For each set $C$ of columns, we have
    $
    \ranksum(\restrict{\Mat}{C})(2) \leq  \ranksum(\restrict{\Mat}{C})(1) $ if and only if $C$ contains $\ell + 1$ and $\sum_{i \in C \setminus \set{\ell + 1}} b_i \leq d$.
\end{lemma}

\begin{proof}

    Consider the following matrix $\Mat \in \M^{2\times (\ell+1)}$:
      
    \begin{center}
    \begin{tabular}{c| c c c c|} 
     \cline{2-5}
     id & $a_1$ & $\ldots$ & $a_n$ & $a_{n+1}$ \\ 
     \cline{2-5}
     1 & 0 & $\ldots$ & 0 & $d + 1$ \\ 
     2 & $b_1$ & $\ldots$ & $b_n$ & 0 \\
    \cline{2-5}
    \end{tabular}
    \end{center}

    It is easy to verify this has the claimed property.
\end{proof}

We can use this theorem to show the hardness of expectation computation (and hence SHAP score computation) whenever we rank by $\scoresum$.

\begin{corollary}[from \cref{thm:sum-prec-hard}]\label{cor:hard-binary}
    For each of the effect functions $e \in \set{\effectf{k\tau}, \effectf{Ham}, \effectf{pos}, \peffectf{k}{top}}$, the problem  $\EXP{\ranksum}{e}$ is $\FP^{\sharpP}$-hard if the input matrix $\Mat$ is encoded in binary.
\end{corollary}

\begin{proof}
    We observe that for $n=2$, expectation computation of the effect functions $e \in \set{\effectf{k\tau}, \effectf{Ham}, \effectf{pos}}$ is equivalent to precedence probability computation, as they are all of the form $a \cdot \prob(\pi(1) > \pi(2)) + b \cdot \prob(\pi(2) > \pi(1))$ with $a \neq b$. 
    
    For $\Mat \in \M^{2\times m}$, $k \geq 1$ and $e = \peffectf{k}{top}$, 
    we consider a matrix $\Mat' \in \M^{(k+1) \times m}$ where $\vect{t'_1} = \vect{t}_1$ and $\vect{t'_2} = \ldots = \vect{t'_{k + 1}} = \vect{t}_2$, so all the tuples $\vect{t'_2}$ up to $\vect{t'_{k+1}}$ are identical. Now, $\vect{t'_1}$ is in the top-$k$ of $\Mat'$ is equivalent to $\vect{t}_1$ being ranked better than $\vect{t}_2$ in $\Mat$. This yields the claim.
\end{proof}

\subsection{Top-k Membership}\label{sec:hardness:top-k}

In \cref{sec:algorithms},  we saw a polynomial-time algorithm for  $\EXP{\rankbyscore{\scoremax}{\desc}}{\peffectf{k}{top}}$ when $k$ is a parameter. In this subsection, we show that this problem is hard for all other ranking functions we consider.

\begin{restatable}{theorem}{thmtopkhard}\label{thm:top-k-hard}
    For each ranking function $r \in \set{\rankbyscore{\scoremax}{\asc}, \ranksum, \ranklex}$ the problem $\EXP{r}{\peffectf{k}{top}}$ is $\FP^{\sharpP}$-hard, when $k$ is a parameter and also when $k$ is part of the input.
    Furthermore, $\EXP{\rankbyscore{\scoremax}{\desc}}{\peffectf{k}{top}}$
        is $\FP^{\sharpP}$-hard if $k$ is  part of the input. 
    These claims remain true if the entries of $\Mat$ are restricted to $\set{0, 1}$ and the distribution $\Pi$ is the uniform distribution on $\set{0,1}^m$.
\end{restatable}

We prove this theorem in \cref{sec:appendix:miss-hardness} using a reduction from counting satisfying assignments to positive CNF formulas. This problem is known to be $\sharpP$-hard~\cite{valiant79complexity}.
Before we continue, we fix some notation for the remainder of this section.
Let $X_1, \ldots, X_m$ be Boolean variables. A vector $\vect u \in \set{0, 1}^m$ defines an assignment $\assign{\vect u}\colon [m] \to \set{\mathsf{true}, \mathsf{false}}$ by interpreting $1$ as $\mathsf{true}$ and $0$ as $\mathsf{false}$. A set $C \subseteq [m]$ defines the assignment $\assign{C} = \assign{\mathbbm 1_C}$, where $\mathbbm 1_C \in \set{0, 1}^m$ is the incidence vector of $C$.  

\begin{example}
We illustrate how the satisfiability of a CNF formula can be modeled using the top-1 membership problem.
Let our formula be $(X_1 \vee X_2 \vee X_4) \wedge (X_1 \vee X_3) \wedge (X_2 \vee X_3 \vee X_4)$.
We construct the matrix shown in \Cref{fig:topk_hardness}, with columns corresponding to formula variables and tuples $\vect{t}_1, \vect{t}_2, \vect{t}_3$ corresponding to clauses.
A value of one indicates that the corresponding variable appears in the clause.
Now, let the column weights be $\{0,1\}$, representing whether a variable is $\mathsf{true}$ or $\mathsf{false}$, and let the ranking function be $\rankbyscore{\scoremax}{\asc}$.
By our construction, $\vect{t}_4$ is top-1 if and only if all clauses are satisfied.\qed
\end{example}

We are now ready to state our main tool for the proof of \cref{thm:top-k-hard}.

\begin{figure}
\begin{center}
\begin{tabular}[b]{c c} 
    \begin{tabular}{c| c c c c|} 
        \cline{2-5}
        $\color{DarkGreen} \vect u$ & $\color{DarkGreen} 1$ & $\color{DarkGreen} 0$ & $\color{DarkGreen} 0$ & $\color{DarkGreen} 1$ \\ 
        \Cline{2-5}{1.0pt}
        $\vect{t}_4$ & \cellcolorblue 0 & \cellcolorblue 0 & \cellcolorblue 0 & \cellcolorblue 0 \\
        $\vect{t}_1$ & 1 & 1 & 0 & 1 \\ 
        $\vect{t}_2$ & 1 & 0 & 1 & 0 \\
        $\vect{t}_3$ & 0 & 1 & 1 & 1 \\
        \cline{2-5}
    \end{tabular}
    &
    \begin{tabular}{c| c c c c|} 
        \cline{2-5}
        $\color{DarkGreen} \vect u$ & $\color{DarkGreen} 0$ & $\color{DarkGreen} 0$ & $\color{DarkGreen} 0$ & $\color{DarkGreen} 1$ \\ 
        \Cline{2-5}{1.0pt}
        $\vect{t}_2$ & 1 & 0 & 1 & 0 \\
        $\vect{t}_4$ & \cellcolorblue 0 & \cellcolorblue 0 & \cellcolorblue 0 & \cellcolorblue 0 \\
        $\vect{t}_1$ & 1 & 1 & 0 & 1 \\ 
        $\vect{t}_3$ & 0 & 1 & 1 & 1 \\
        \cline{2-5}
    \end{tabular}
\end{tabular}
\end{center}
\caption{\Cref{ex:max_topk}: Tuple $\vect{t}_4$ is top-1 under $\rankbyscore{\scoremax}{\asc}$ iff $\vect u$ gives a satisfying assignment for the CNF formula $(X_1 \vee X_2 \vee X_4) \allowbreak \wedge \allowbreak (X_1 \vee X_3) \allowbreak \wedge \allowbreak (X_2 \vee X_3 \vee X_4)$.
The left is a satisfying assignment, while the right one is non-satisfying.}
\label{fig:topk_hardness}
\end{figure}

\begin{restatable}{lemma}{lemtopkhardmatrix}\label{lem:top-k-hard-matrix}
    Let $\phi = \bigwedge_{i = 1}^\ell D_i$ be a positive CNF formula with variable set $X_1, \ldots, X_m$ and clauses $D_i = \bigvee_{j = 1}^{r_i} X_{j_i}$. Furthermore, let $k \in \mathbb N$. Then, there is a matrix $\Mat \in \M^{(\ell + k) \times m}$ with entries in $\set{0,1}$ and the following property: For each set $C$ of columns and each ranking function $r \in \set{\rankbyscore{\scoremax}{\asc}, \ranksum, \ranklex}$, we have
    $
    r(\restrict{\Mat}{C})(k + \ell) \leq  k$ if and only if the assignment $\assign{C}$ models $\phi$.
\end{restatable}

\begin{remark}
    Using \cite[Corollary~10]{JMLR:ArenasBarcelo23SHAP}, the previous lemma directly shows that we cannot approximate SHAP scores for $r \in \set{\rankbyscore{\scoremax}{\asc}, \ranksum, \ranklex}$ and $e = \peffectf{k}{top}$ unless $\RP = \NP$. 
\end{remark}

\subsection{Maximum Displacement Distance}\label{sec:hardness:md}

In this section, we prove that expectation computation is a hard problem if we measure the effect using the maximum displacement distance $\distmaxdispl(\pi_1, \pi_2) = \max_{1 \leq i \leq n} \size{\pi_1(i) - \pi_2(i)}$.

\begin{restatable}{theorem}{thmmdhard}\label{thm:md-hard}
For $e = \effmd$ and $r \in \set{\rankbyscore{\scoremax}{\asc}, \rankbyscore{\scoremax}{\desc}, \ranksum, \ranklex}$, the problem $\EXP{r}{e}$ is $\FP^{\sharpP}$-hard. This remains true if the entries of $\Mat$ are restricted to $\set{0, 1}$ and the distribution $\Pi$ is the uniform distribution on $\set{0,1}^m$.
\end{restatable}

We prove this theorem in \cref{sec:appendix:miss-hardness}  via a reduction from counting satisfying assignments to positive CNF formulas. The proof is similar to the one of \cref{thm:top-k-hard}, but with a small twist. Since $\distmaxdispl$ is a global measure, we cannot define a matrix $\Mat$ where we separate satisfying assignments from non-satisfying ones. Instead, we will define two matrices $\Mat_1$ and $\Mat_2$ in a way that the difference of the effect functions allows this distinction.

\begin{restatable}{lemma}{lemmdhardmatrices}\label{lem:md-hard-matrices}
     Let $\phi = \bigwedge_{i = 1}^\ell D_i$ be a positive CNF formula with variable set $X_1, \ldots, X_m$ and clauses $D_i = \bigvee_{j = 1}^{r_i} X_{j_i}$. Then, there exist two matrices $\Mat_1 \in \M^{(2\ell + 1) \times m}$ and $\Mat_2 \in \M^{(2\ell + 2) \times m}$ with entries in $\set{0,1}$ and the following property: For each set $C$ of columns and each ranking function $r \in \set{\rankbyscore{\scoremax}{\asc}, \ranksum, \ranklex}$, we have
    \begin{equation}\label{eqn:diff-md}
        \distmaxdispl\big(r(\Mat_2), r(\restrict{\Mat_2}{C})\big) - \distmaxdispl\big(r(\Mat_1), r(\restrict{\Mat_1}{C})\big) = \begin{cases}
        0,\text{ if } \assign{C}\models \phi \\
        1,\text{ otherwise.}
        \end{cases}
    \end{equation}
\end{restatable}

\begin{proof}[Proof sketch of \cref{lem:md-hard-matrices}]
The construction of the matrices is similar to \Cref{ex:max_topk,lem:top-k-hard-matrix} with more copies of the all-zero tuple, which occupy the top positions whenever the formula is satisfied, yielding zero displacement.
Matrix $\Mat_2$ contains one additional copy of the all-zero tuple compared to $\Mat_1$,
so the two matrices behave similarly, but the maximum displacement in $\Mat_2$ is higher than that of $\Mat_1$ by exactly 1 whenever the formula is not satisfied.  
\end{proof}

\section{Exact Computation of the Shapley Value}\label{sec:shapley}

In this section, we turn our attention to the importance of columns via the Shapley value instead of the importance of scoring parameters via the SHAP score. We will first show that all our algorithms from \cref{sec:algorithms} for SHAP score computation can be used to determine Shapley values and then complete the picture by extending our results from \cref{sec:hardness}.

\begin{restatable}{theorem}{thmcomplexityshapley}\label{thm:complexity-shapley}
    For the ranking functions $r$ and effect functions $e$ that we introduced in \cref{sec:problems}, 
    computational complexity of $\PShapley{r}{e}$ is the same as the computational complexity of $\PSHAP{r}{e}$ given in \cref{table:complexity}. 
\end{restatable}

\subsection{Algorithms}
Gilad et al.~\cite{gilad2024importanceparametersdatabasequeries} give a simple algorithm based on interpolation that computes Shapley values in polynomial time when having access to an oracle to a special instance of SHAP score computation called $\mathsf{binarySHAP}$. Translated into our setting, their result is the following:

\begin{theorem}[\protect{\cite[Theorem~6.15]{gilad2024importanceparametersdatabasequeries}}]\label{thm:Shapley-from-SHAP}
    Let $r$ be a ranking function and $e$ be an effect function with the property that for each set $C$ of columns, we have $e(r(\restrict{\Mat}{C})) = e(r(\Mat \circ \mathbbm 1_C))$. Then, we can solve $\PShapley{r}{e}$ in polynomial time with $m$ oracle calls to $\PSHAP{r}{e}$.
\end{theorem}

While the condition of the theorem is always true for ranking by $\scoresum$ and $\lex$, it does not necessarily hold for $\scoremax$ ranking if $\Mat$ contains negative numbers. Fortunately, we can make $\Mat$ non-negative by adding the minimal value in $\Mat$ to each entry without changing the rankings of each set of columns. Hence, we can use the previous theorem to show that all our algorithms carry other to Shapley value computation.

\subsection{Intractable Cases}\label{sec:shapley-hard}

To show hardness results from \cref{sec:hardness}, we reduced the counting knapsack problem (\cref{lem:sum-hard-matrix}) and the problem of counting satisfying assignments to positive CNF formulas (\cref{lem:md-hard-matrices,lem:top-k-hard-matrix}) to SHAP score computation via the equivalence to expectation computation (\cref{thm:shap-expectation}). Since this equivalence does not hold for Shapley value computation in general, we need another way approach to extend our results. Instead, we will define two auxiliary games and show hardness of Shapley value computation for each of them. The proofs are given in \cref{sec:appendix:shapley}. They follow the technique that is used for most hardness proofs of Shapley value computation~\cite{DBLP:journals/lmcs/LivshitsBKS21, DBLP:conf/icdt/GroheK0S24}: For a given instance to a game, define a set of related instances and use their Shapley values to count valid knapsack sets or satisfying assignments, respectively, via interpolation.

The first game is the \emph{knapsack game}: For natural numbers $b_1, \ldots, b_\ell$ and $b_{\ell + 1} = d$ given in binary, consider the game with player set $[\ell + 1]$ and valuation function 
\[
\nuknap(C) \defeq \begin{cases}
    1 & \text{ if } \ell + 1 \in C \text{ and } \sum_{i \in C\setminus\set{\ell}} b_i \leq d \,\mbox{;} \\
    0 & \text{ otherwise.}    
\end{cases} 
\]

\begin{restatable}{lemma}{lemknapsackgame}\label{lem:game-knapsack}
   Shapley value computation for the players of the knapsack game is $\FP^{\sharpP}$-hard.
\end{restatable}

The second game is the \emph{positive CNF game}: Given a positive CNF formula $\phi$ with variable set $X_1,\ldots,X_m$, consider the game with player set $[m]$ and valuation function
\[
\nucnf(C) \defeq \begin{cases}
    1 & \text{ if } \assign{C} \models \phi\,; \\
    0 & \text{ otherwise.}    
\end{cases} 
\]

\begin{lemma}\label{lem:game-cnf}
   Shapley value computation for the players of the positive CNF game is $\FP^{\sharpP}$-hard.
\end{lemma}

With these lemmas and the constructions from 
\cref{sec:hardness}, we can easily prove \cref{thm:complexity-shapley}. The details are presented in \cref{sec:appendix:shapley}.

\section{Conclusions}\label{sec:conclusions}
We studied the computational complexity of measuring the importance of parameter choices in ranking functions. 
Specifically, we investigated the SHAP score of parameters under a collection of specific (yet basic) ranking functions and effect functions. 
For that, we studied the complexity of calculating the expected effect over random parameter values, focusing on finite and 
probabilistically independent distributions. We also studied the related problem of computing the Shapley value of columns for measuring their contribution to the ranking. 

We view this work as a first step in the rigorous analysis of the computational complexity of function-based explanations for rankings. 
As such, many directions are left for future investigation. Most importantly, we would like to generalize our results into broad \e{classes} of ranking functions and effect functions for which we can devise efficient algorithms, rather than considering specific cases separately.  In particular, we would like to be able to analyze the implication of the different choices made in popular ranking schemes; examples include sports ranking such as WBSC\footnote{\url{https://www.wbsc.org/en/rankings}} and FIFA/Coca-Cola World Ranking,\footnote{\url{https://inside.fifa.com/fifa-world-ranking}} which have the shape of workflow diagrams with linear combinations of features within, and the Computer Science Rankings\footnote{\url{https://csrankings.org/}} that is based on a pre-determined choice of publication venues. We would also like to extend our analysis to other classes of probability distributions, including uniform distributions over intervals, and to characterize the cases in which we can efficiently approximate SHAP scores multiplicatively. Moreover, we would like to understand the complexity of general types of effect functions that include alternatives to those studied here, such as Spearman's footrule and the Cayley distance. 

Finally, an important  direction for future research is to study the practicality of the framework in real-life scenarios, and particularly understand how well the attribution functions of SHAP and Shapley capture people's intuition on the contribution of the components of a ranking function.

\bibliographystyle{plainurl}
\bibliography{refs}

\newpage
\appendix

\label{appendix:beginning}

\hide{
\section{Nomenclature}
\label{sec:appendix:nomenclature}

\begin{table}[h]
\centering
\small
\begin{tabularx}{\linewidth}{@{\hspace{0pt}} >{$}l<{$} @{\hspace{2mm}}X@{}} %
\hline
\textrm{Symbol}		& Definition 	\\
\hline
    \hline
	r				& ranking function \\
\hline
\end{tabularx}
\end{table}
}

\section{Proofs Omitted from \cref{sec:algorithms}}\label{sec:appendix:miss-algorithms}

In this section, we give the proofs that we omitted in \cref{sec:algorithms}.

\thmmaxlexktauposptime*
\begin{proof}[Proof of \cref{thm:max-lex-ktau-pos-ptime}]
For both measures, we write the expected effect as the expectation of a sum of indicator variables for tuple precedence and then use linearity of expectation.

With $\pi$ being $r(\Mat\circ \vect u)$, for the expected rank $\effectf{pos}$ of a tuple with row index $i$ we have:
\[
\expectation_{\vect u\sim\Pi}[\effectf{pos}(i, \pi)] = \expectation_{\vect u\sim\Pi}\Big[\sum_{j \in [n]} \one_{\pi(i) > \pi(j)}\Big]
= \sum_{j \in [n]} \expectation_{\vect u\sim\Pi} [\one_{\pi(i) > \pi(j)}]
= \sum_{j \in [n]} \prob_{\vect u\sim\Pi} [\pi(i) > \pi(j)].
\]
Therefore, for all ranking functions $r$ for which the $\PREC{r}$ problem is solvable in polynomial time according to \Cref{thm:prec_algs},
the problem $\EXP{r}{\effectf{pos}}$ is also in polynomial time by solving the former $\bigO(n)$ times.

Similarly, for the expected Kendall's tau distance $\effectf{k\tau}$ we have:
\begin{align*}
    \expectation_{\vect u\sim\Pi}[\effectf{k\tau}(\pi)] &= 
\expectation_{\vect u\sim\Pi}\big[ \!\!\! \sum_{1\leq i<j\leq n} \!\!\! \one_{\pi(i)>\pi(j)}\big]
\\ &= \!\!\! \sum_{1\leq i<j\leq n} \!\!\! \expectation_{\vect u\sim\Pi} [\one_{\pi(i)>\pi(j)}]
= \!\!\! \sum_{1\leq i<j\leq n} \!\!\! \prob_{\vect u\sim\Pi} [\pi(i) > \pi(j)].
\end{align*}

Thus, for $\EXP{r}{\effectf{k\tau}}$, we solve the $\PREC{r}$ problem $\bigO(n^2)$ times, which remains polynomial.
\end{proof}

\thmmaxtopkptime*
\begin{proof}[Proof of \cref{thm:max-topk-ptime}]
Let the tuple of interest be $\vect{t}_i$ with row index $i$.
The task is to compute the probability of the event 
\[
T_{n-k}\defeq \text{ tuple } \vect{t}_i \text{ beats (i.e., appears before) at least $n-k$ tuples in the ranking.}
\]
Using inclusion-exclusion~\cite{feller1968book}, this can be expressed as follows:
\[
\prob[T_{n-k}] \!=\! \prob[W_{n-k}] - {n \! - \! k \choose 1} \prob[W_{n-k+1}] + {n \! - \! k \choose 2} \prob[W_{n-k+2}] - \ldots \pm {n \! - \! 2 \choose n \! - \! k \! - \! 1} \prob[W_{n-1}]
\]
with $\prob[W_{d}] = \sum_{S \subseteq n, |S|=d}  \prob[\bigwedge_{\ell \in S} (\pi(i) < \pi(\ell))]$.

The number of subsets $S$ is ${n \choose n-k}$, which is polynomial because $k$ is constant.
Thus, it suffices to compute in polynomial time the probability of each event $B_{S}\defeq$ $\bigwedge_{\ell \in S} (\pi(i) < \pi(\ell))$,
which corresponds to tuple $\vect{t}_i$ beating all tuples with row indexes in $S$.

First, assume that tuple $\vect{t}_i$ is after the tuples of $S$ in the tie-breaking mechanism.
From the tuples in $S$, we construct a single competitor $\vect{t'}$ that contains the largest value per column.
We argue that $\vect{t}_i$ is before $\vect{t'}$ in the ranking if and only if it is before $\vect{t'}$.
Indeed:
\[
\max_{j \in m} (\vect{t}_i[j]) > \max_{\ell \in S} (\max_{j \in m} (\vect{t}_\ell[j]))
\Leftrightarrow 
\max_{j \in m} (\vect{t}_i[j]) > (\max_{j \in m} \max_{\ell \in S} (\vect{t}_\ell[j]))
\Leftrightarrow 
\max_{j \in m} (\vect{t}_i[j]) > \max_{j \in m} (\vect{t'}[j]). 
\]
To handle the case where some tuples in $S$ are above $\vect{t}_i$ in the tie-breaking while others are below,
we split them into these two groups and create two competitors $\vect{t'}$ and $\vect{t''}$.
We invoke the algorithm for $\PREC{\rankbyscore{\scoremax}{\desc}}$ separately for $(\vect{t}_i, \vect{t'})$ and $(\vect{t}_i, \vect{t''})$, with a different assumption for the resolution of ties.
By \Cref{obs:list_of_weights}, we can obtain the lists of valid weights for each invocation and then intersect them to obtain the final list and the corresponding probability is straightforward to compute.

Overall, our algorithm needs to solve the $\PREC{\rankbyscore{\scoremax}{\desc}}$ problem $\bigO(k n^k)$ times.
\end{proof}

\subsection{Proofs Omitted from \cref{sec:bounded_dim}}

\thmboundeddim*

\begin{proof}
First, suppose that $m$ is $\bigO(1)$.
The probability space for possible $\vect u$ weights is $\bigO(|\Pi|^m)$, which is now polynomial, so we can enumerate all possible weight assignments.
For each assignment, we sort the table by $r$ and compute the effect $e$.
The expectation of the effect is then calculated as a weighted sum over these cases.

Second, suppose that $n$ is $\bigO(1)$.
While the probability space is now exponentially large, the number of possible permutations of the rows of the $\Mat$ matrix is $n!$, which is $\bigO(1)$.
We construct all possible permutations and, for each one, we will determine its probability in polynomial time, allowing us to then compute the expected effect.

Let $\pi$ be one of the permutations. 
For ranking functions $\rankbyscore{\scoremax}{\asc}, \rankbyscore{\scoremax}{\desc},\ranklex$, we consider all comparisons of consecutive tuples $\pi(i) < \pi(i+1)$ for $i\in[n-1]$.
These comparisons are sufficient to describe the permutation $\pi$.
For each comparison, we compute the satisfying weights by \Cref{thm:prec_algs,obs:list_of_weights} and then intersect them to obtain the probability of $\pi$.

Now, let our ranking function be $\ranksum$ with the additional assumption that the matrix values and weights in $\Pi$ are integers and given in unary.
We modify the dynamic program of \Cref{thm:prec_algs} so that it computes the probability of the entire permutation.
Suppose the tuples $\vect{t}_1, \ldots, \vect{t}_n$ are indexed by the order of the permutation $\pi$.
Also, let $S_1$ be the set of tuples $\vect{t}_i, i\in[n]$ that appear before $\vect{t}_{i+1}$ in the tie-breaking scheme, and $S_2$ be the rest of them.
\begin{align}
&R[j, s_1, \ldots, s_{n-1}] = \\
&\notag \qquad \sum_{v_j \in \supp(\Pi_j)} \prob[u_j = v_j] \cdot R[j + 1, 
s_1 - v_j (\vect{t}_2[j] - \vect{t}_1[j],
\ldots,
s_{n-1} - v_j (\vect{t}_n[j] - \vect{t}_{n-1}[j]
)]\\
& R[m+1, s_1, \ldots, s_{n-1}] = 
\begin{cases} 
            1,  \text{ if } \forall i: \vect{t}_i \in S_1 \Rightarrow s_i \geq 0  \text{ and } \forall i: \vect{t}_i \in S_2 \Rightarrow s_i > 0\\
	  	0, \text{ otherwise.} 
\end{cases}        
\end{align}

The final answer is $R[1,0,\ldots,0]$.
The running time compared to the previous dynamic program is greater by a factor of $\bigO(n)$, which is a constant.

\end{proof}

\section{Proofs Omitted from \cref{sec:hardness}}\label{sec:appendix:miss-hardness}
In this section, we give the proofs that we omitted in \cref{sec:hardness}.

\subsection{Proofs Omitted from \cref{sec:hardness:sum}}
\thmsumprechard*

\begin{proof}[Proof of \cref{thm:sum-prec-hard}]
    Let $\Mat$ be the matrix from \cref{lem:sum-hard-matrix} and $\Pi$ be the uniform distribution on $\set{0, 1}^{\ell+1}$. By \cref{lem:sum-hard-matrix}, a weight vector $\vect u=(u_1, \ldots, u_{\ell + 1}) \in \set{0,1}^{\ell + 1}$ 
    fulfills $\scoresum(\vect{t}_2 \circ \vect u) \leq \scoresum(\vect{t}_1 \circ \vect u)$ if and only if $u_{\ell + 1} = 1$ and the set $S = \set{1 \leq i \leq \ell \, \mid \, u_i = 1}$ satisfies $\sum_{i \in S} b_i \leq d$. Hence, the answer to the counting knapsack problem is exactly the precedence probability of $\vect{t}_2$ and $\vect{t}_1$ times $2^{\ell+1}$. 
\end{proof}
\subsection{Proofs Omitted from \cref{sec:hardness:top-k}}

\thmtopkhard*

\begin{proof}[Proof of \cref{thm:top-k-hard}]
    We first observe that the claim for $r \in \set{\rankbyscore{\scoremax}{\asc}, \ranksum, \ranklex}$ follows immediately from \cref{lem:top-k-hard-matrix}. Let $\Mat$ be the matrix from \cref{lem:top-k-hard-matrix} and $\Pi$ be the uniform distribution on $\set{0, 1}^{m}$. Since for each $\vect u \in \set{0,1}^m$, the rank of $\vect{t}_{\ell + k}$ in $r(\Mat \circ \vect u)$ is $\leq k$ if and only if $\assign{\vect u} \models \phi$, the number of satisfying assignments of $\phi$ is $2^m$ times the probability of $\vect{t}_{k + \ell}$ having rank $\leq k$.

    For $r = \rankbyscore{\scoremax}{\desc}$, we observe
    \begin{align*}
\prob(\rankbyscore{\scoremax}{\desc}(\vect{t}) \leq k) 
=\prob(\rankbyscore{\scoremax}{\asc}(\vect{t}) \geq n - k) = 1- \prob(\rankbyscore{\scoremax}{\asc}(\vect{t}) \leq n - k -1),
\end{align*}
    so it does not matter if we use ascending or descending order when $k$ is part of the input.
\end{proof}

\lemtopkhardmatrix*

\begin{proof}
    We define $\Mat \in \M^{(\ell + k) \times m}$ as follows. The first $k-1$ tuples have all entries equal to $0$. The next $\ell$ tuples
    $\vect{t}_k, \ldots, \vect{t}_{k + \ell - 1} $ are the incidence tuples of $D_1, \ldots, D_\ell$; that is, 
	\[
	  \vect{t}_{k - 1 + i}[j] \defeq \begin{cases}
	  	1, \text{ if } X_j \text{ appears in }	D_i \\
	  	0, \text{ otherwise.} 
	  	  \end{cases}
	\] 
	The last tuple $\vect{t}_{\ell + k}$ again has all entries equal to $0$.

    Let $C$ be a set of columns and consider $\restrict{\Mat}{C}$. In $\restrict{\Mat}{C}$, there are two kinds of tuples: zero tuples, where all entries are $0$, and non-zero tuples, which contain at least one $1$. Zero-tuples are pairwise incomparable in all three ranking functions, so they will stay in their respective order. Each non-zero tuple gets a higher rank than all the zero-tuples as its scores are higher and also $\mathsf{Lex}$ ranks it higher in the column of its first non-zero entry. 

    Since $\vect{t}_{\ell + k}$ has the highest index and all the $k$ tuples $\vect{t}_{1}, \ldots, \vect{t}_{k-1}$ and $\vect{t}_{k + \ell}$ are zero tuples with respect to every column set, $\vect{t}_{\ell + k}$ gets rank $\leq k$ if and only if none of the tuples $\vect{t}_k, \ldots, \vect{t}_{k + \ell - 1}$ is a zero tuple w.r.t. $C$. Finally, we observe that for $1\leq i \leq \ell$, the tuple $\vect{t}_{k -1 + i}$ is a non-zero tuple if and only if $\assign{C} \models D_i$. Together, this shows that $\vect{t}_{\ell + k}$ has rank $\leq k$ if and only if $\assign{C} \models \phi$.
\end{proof}

\subsection{Proofs Omitted from \cref{sec:hardness:md}}
\lemmdhardmatrices*
\begin{proof}
    For $k \in \set{1,2}$, let $\Mat_k \in \M^{(2\ell + k) \times m}$ be the matrix where the first $\ell$ tuples $\vect{t}_1, \ldots, \vect{t}_\ell$ are the incidence tuples of the clauses $D_1, \ldots, D_\ell$. For simplicity, we assume that these tuples are ordered as in $r(\Mat_k)$. This is always the case for $\scoremax$, but for $\scoresum$ and $\mathsf{Lex}$, it means that the clauses $D_i$ of $\phi$ are ordered by their number of variables or by lexicographic order of their incidence vectors, respectively. The remaining $\ell + k$ tuples have all entries equal to zero. We claim that $\Mat_2$ and $\Mat_1$ meet \cref{eqn:diff-md}.

    Let $C\subseteq [m]$ be a column set. Similar to the proof of \cref{lem:top-k-hard-matrix}, we observe the following:
    If $\assign{C} \models \phi$, then the all-zero tuples $\vect{t}_{\ell + 1}, \ldots, \vect{t}_{2\ell + k}$ are ranked first in $r(\restrict{\Mat_k}{C})$ and the clause tuples $\vect{t}_{1}, \ldots \vect{t}_{\ell}$ appear later in the ranking. 
    This holds in particular for $\Mat_k = \restrict{\Mat_k}{[m]}$. If $\assign{C} \nvDash \phi$, then the tuples corresponding to non-satisfied clauses are ranked first, then the all-zero tuples and finally the tuples corresponding to satisfied clauses.

    We now refine this analysis to prove \cref{eqn:diff-md}. Let us first assume that $\assign{C} \models \phi$. Then, the all-zero tuples are ranked first in both $r(\Mat_k)$ and $r(\restrict{\Mat_k}{C})$, so their displacement is $0$. The displacement of the clause tuples does not depend on the number of all-zero-tuples, so \cref{eqn:diff-md} holds in that case. Now, assume that $\assign{C} \nvDash \phi$ and let $\set{i_1, \ldots, i_s}$ be the indices of the non-satisfied clauses. Since we assumed that these indices also represent the relative order of these tuples in $r(\Mat_k)$, we find $r(\Mat_k)(i_j) = \ell + k + i_j$ and $r(\restrict{\Mat_k}{C})(i_j) = j$, so the displacement of each $i_j$ is at least $\ell + k$ and the displacements in $\Mat_1$ and $\Mat_2$ differ by $1$. The rank of each all-zero tuple increases by the number of non-satisfied clauses from $r(\Mat)$ to $r(\restrict{\Mat_k}{C})(i_j) = j$, so their displacement is at most $\ell$ and it is the same in $\Mat_1$ and $\Mat_2$. The tuples corresponding to satisfied clauses stay within the last $\ell$ ranks of the matrices, so their displacement is also $\leq \ell$ and it again does not depend on the number of all-zero tuples. Hence, the maximum displacement is attained for one of the tuples corresponding to a non-satisfied clause and its value in $\Mat_1$ and $\Mat_2$ differs by $1$. This shows the claim.
\end{proof}

\thmmdhard*

\begin{proof}
    We first observe that claim for $r \in \set{\rankbyscore{\scoremax}{\asc}, \ranksum, \ranklex}$ follows immediately from \cref{lem:top-k-hard-matrix}. Let $\Mat_1$ and $\Mat_2$ be the matrices from \cref{lem:md-hard-matrices} and $\Pi$ be the uniform distribution on $\set{0, 1}^{m}$. Since for each $\vect u \in \set{0,1}^m$, the difference of the effect function on the two matrices is $0$ if $\vect u \models \phi$ and $1$ otherwise,
    the number of satisfying assignments of $\phi$ is $2^m$ minus $2^m$ times the difference of the expected values of the .

    For $r = \rankbyscore{\scoremax}{\desc}$, we observe
    that $\distmaxdispl$ is invariant under reversing both orders,
    so it does not matter if we use ascending order or descending order with reversed tie-breaking.

\end{proof}

\def\secshapleyref{\Cref{sec:shapley}}

\section{Proofs Omitted from \secshapleyref}\label{sec:appendix:shapley}

\thmcomplexityshapley*
\begin{proof}
    By \cref{thm:Shapley-from-SHAP}, all polynomial time algorithms for $\PSHAP{r}{e}$ can be used to compute $\PShapley{r}{e}$ in polynomial time.
    The hardness results we need to show are, that the problem $\PShapley{r}{e}$ is $\FP^{\sharpP}$-hard for the following combinations of $r$ and $e$:
    \begin{enumerate}[label=(\alph*)]
        \item\label{itm:shapley-sum} $r = \ranksum$ and $e \in \set{\effectf{k\tau}, \effectf{Ham}, \effectf{pos}, \peffectf{k}{top}}$, if the input matrix $\Mat$ is encoded in binary,
        \item\label{itm:shapley-top-k-p} $r \in \set{\rankbyscore{\scoremax}{\asc}, \ranksum, \ranklex}$ and $e = \peffectf{k}{top}$, where $k$ is a parameter to the problem,
        \item\label{itm:shapley-top-k-i} $r = \rankbyscore{\scoremax}{\desc}$ and $e = \peffectf{k}{top}$, when $k$ is part of the input, and
        \item\label{itm:shapley-md} $r \in \set{\rankbyscore{\scoremax}{\asc}, \rankbyscore{\scoremax}{\desc}, \ranksum, \ranklex}$ and $e = \effmd$.
    \end{enumerate}
    Claim \ref{itm:shapley-sum} follows directly from \cref{lem:game-knapsack} and \cref{lem:sum-hard-matrix} using the reduction from $\PREC{r}$ in \cref{cor:hard-binary}. Claims \ref{itm:shapley-top-k-p} and \ref{itm:shapley-top-k-i} follow from \cref{lem:game-cnf} and \cref{lem:top-k-hard-matrix} using the equivalence of $\scoremax\asc$ and $\scoremax\desc$ when $k$ is part of the input that we established in the proof of \cref{thm:top-k-hard}.
    Finally, claim \ref{itm:shapley-md} follows from \cref{lem:game-cnf} and \cref{lem:md-hard-matrices} using the facts that $\distmaxdispl$ is invariant under reversing both permutations and that ranking in descending order with reversed tie-breaking order yields the reversed order of ranking in ascending order. 
\end{proof}

\lemknapsackgame*

\begin{proof}
    We show that Shapley value computation in the knapsack game is $\FP^{\sharpP}$-hard by reducing from the counting knapsack problem. 
    
    Let $I = (b_1, \ldots, b_\ell; d)$ be an input the knapsack game. We are interested in the Shapley value of player $\ell + 1$, that belongs to the knapsack capacity $d$.
    For a set $C \subseteq [\ell]$, we observe that $\nuknap(C) = 0$ and
    \begin{equation}\label{eqn:nuknap}
        \nuknap(C \cup \set{\ell + 1}) - \nuknap(C) = \begin{cases}
        1,\text{ if } \sum_{i \in C} b_i \leq d \\
        0,\text{ otherwise.} 
        \end{cases}   
    \end{equation}
    Let $s_k$ be the number of valid item sets of size $k$; that is 
    \[s_k = \size[\Big]{\set[\big]{C \subseteq [\ell] \, \mid \, \size{C} = k \, \wedge \, \sum_{i \in C} b_i \leq d}}.\]
    With this notation and \cref{eqn:nuknap}, we can express the Shapley value of $\ell + 1$ by
    \[
    \Shapleyvalue(\ell + 1, I, \nuknap) = \sum_{k = 0}^{\ell} \frac{k!(\ell - k)!}{(\ell + 1)!} s_k.
    \]
    
    For $0 \leq j \leq \ell$ we define new instances $I_j$ where we add $j$ items with weights $d+1$, so $I_j = (b_1, \ldots, b_\ell, d + 1, \ldots, d + 1; d)$. Since each of the added weights is larger than the knapsack capacity $d$, none of them can be part of a valid item set, so 
    \[
    \Shapleyvalue(\ell + j + 1, I_j, \nuknap) = \sum_{k = 0}^{\ell} \frac{k!(\ell + j - k)!}{(\ell + j + 1)!} s_k.
    \]
    This yields the following system of linear equations:
    \[\begin{pmatrix}
    \ell! & (\ell-1)! & \cdots & 0! \\
    (\ell+1)! & \ell! & \cdots & 1! \\
    \vdots & \vdots & \ddots & \vdots \\
    (2\ell)! & (2\ell-1)! & \cdots & \ell!
    \end{pmatrix} \cdot \begin{pmatrix}
    0! s_0\\
    1! s_1\\
    \vdots \\
    \ell! s_\ell
\end{pmatrix} = \begin{pmatrix}
    (\ell + 1)!\Shapleyvalue(\ell + 1, I_0, \nuknap) \\
    (\ell + 2)!\Shapleyvalue(\ell + 2, I_1, \nuknap)  \\
    \vdots \\
    (2\ell + 1)!\Shapleyvalue(2\ell + 1, I_\ell, \nuknap) 
\end{pmatrix}
\]
This system of linear equations has full rank as the matrix on the left is invertible as shown by \cite[proof of Theorem 1.1.]{bacher2002determinants}. Hence, we can determine the number of valid item sets by solving this system and returning $s_0 + \ldots + s_\ell$.
\end{proof}

\begin{proof}[Proof of \cref{lem:game-cnf}]
    We show that Shapley value computation in the positive CNF game is $\FP^{\sharpP}$-hard by reducing from the problem of counting satisfying solutions to a CNF formula.

    For positive CNF formula $\phi = \bigwedge_{i = 1}^\ell D_i$ with variable set $X_1, \ldots, X_m$ and clauses $D_i = \bigvee_{j = 1}^{r_i} X_{j_i}$, we define a new CNF formula $\hat\phi$ from $\phi$ by introducing a new variable $X_{m + 1}$ and adding this variable to each of the clauses. Then, every coalition that contains $m+1$ satisfies the formula and we obtain for a coalition $C \subseteq [m]$:
    \begin{equation}\label{eqn:nucnf}
          \nucnf(C \cup \set{X_{m + 1}}) - \nucnf(C) = \begin{cases}
        0,\text{ if } \assign{C} \models \phi \\
        1,\text{ otherwise.} 
        \end{cases} 
    \end{equation}

    Let $\bar s_k$ be the number of variable assignments assignments $\alpha \colon [m] \to \set{\mathsf{false}, \mathsf{true}}$ that do \emph{not} model $\phi$.
    With this notation and \cref{eqn:nucnf}, we can express the Shapley value of $m + 1$ by
    \[
    \Shapleyvalue(m + 1, \hat\phi, \nucnf) = \sum_{k = 0}^{m} \frac{k!(m - k)!}{(m + 1)!} \bar s_k.
    \]
    For $0 \leq j \leq m$, we define new CNF formulas $\phi_j = \bigwedge_{i = 1}^{\ell} D^{(j)}_i$ by introducing $j$ new variables $X_{m+1}, \ldots, X_{m + j}$ and adding each of them to all the clauses $D_i$. Now, a variable assignment $\alpha' \colon [m + j] \to \set{\mathsf{false}, \mathsf{true}}$ does not model $\phi_j$ if $\restrict{\alpha'}{[m]} \nvDash \phi$ and $\alpha'(i) = \mathsf{false}$ for all $m + 1 \leq i \leq m + j$, so 
    \[
    \Shapleyvalue(m + j + 1, \hat\phi_j, \nucnf) = \sum_{k = 0}^{m} \frac{k!(m + j - k)!}{(m + j + 1)!} \bar s_k
    \]
     This yields the following system of linear equations:
    \[\begin{pmatrix}
    m! & (m-1)! & \cdots & 0! \\
    (m+1)! & m! & \cdots & 1! \\
    \vdots & \vdots & \ddots & \vdots \\
    (2m)! & (2m-1)! & \cdots & m!
    \end{pmatrix} \cdot \begin{pmatrix}
    0! \bar s_0\\
    1! \bar s_1\\
    \vdots \\
    m! \bar s_\ell
\end{pmatrix} = \begin{pmatrix}
    (m + 1)!\Shapleyvalue(m + 1, \hat\phi_0, \nucnf) \\
    (m + 2)!\Shapleyvalue(m + 2, \hat\phi_1, \nucnf)  \\
    \vdots \\
    (2m + 1)!\Shapleyvalue(2m+ 1, \hat\phi_m, \nucnf) 
\end{pmatrix}
\]
This system of linear equations has full rank as the matrix on the left is invertible as shown by \cite[proof of Theorem 1.1.]{bacher2002determinants}. Hence, we can determine the number of satisfying assignments of $\phi$ by solving this system and returning $2^m - \bar s_0 - \ldots - \bar s_m$.
\end{proof}

\section{The Top-k Perspective}\label{sec:appendix:top-k}

We show that the results that we have proved for the top-$k$ membership $\peffectf{k}{top}$ effect function (see \Cref{table:complexity})
also apply to the 
top-$k$ (symmetric) difference $\peffectf{k}{\mathrm{\Delta}}$
and any-change $\peffectf{k}{any}$ effect functions.

\begin{theorem}
Let $r \in \set{\ranksum, \rankbyscore{\scoremax}{\asc}, \rankbyscore{\scoremax}{\desc},\ranklex}$
and 
$e \in \set{\peffectf{k}{\mathrm{\Delta}}, \peffectf{k}{any}}$.
The problem $\EXP{r}{e}$ is solvable in polynomial time if $r=\rankbyscore{\scoremax}{\desc}$ and $k$ is a fixed parameter.
Otherwise, it is $\FP^{\sharpP}$-hard, and it remains $\FP^{\sharpP}$-hard for 
$r \in \set{\ranksum, \rankbyscore{\scoremax}{\asc},\ranklex}$
even if $k$ is a fixed parameter.
\end{theorem}
\begin{proof}
Let $T_0$ be the top-$k$ set in the ranking of the whole matrix $\Mat$ and $T$ be the top-$k$ set in another permutation induced by the ranking function applied to a set $C$ of columns, so our two effect functions are defined as
$\peffectf{k}{\mathrm{\Delta}}(\pi) \defeq \size{T\cup T_0}-\size{T\cap T_0}$
and
$\peffectf{k}{any}(\pi) \defeq \one_{T\neq T_0}$.
Instead of the expected symmetric difference, we will show the theorem for the
expected intersection size $\peffectf{k}{\cap}(\pi) \defeq \size{T\cap T_0}$, since 
$\size{T\cup T_0}-\size{T\cap T_0} = 2k - 2\size{T\cap T_0}$,
which means that we can derive one from the other.

Starting from the negative side, we observe that the matrix $\Mat$ defined in \cref{lem:top-k-hard-matrix} from a positive CNF formula $\phi$ has the following property: If $\assign{C} \models \phi$, then $T_0 = T$, and otherwise $\size{T_0 \cap T} = k - 1$. From the arguments from the proof of \Cref{thm:top-k-hard}, we can infer $\FP^{\sharpP}$-hardness of $\EXP{r}{\peffectf{k}{\cap}}$ and $\EXP{r}{\peffectf{k}{any}}$ for the combinations of $r$ and $e$ stated in the theorem. 

It remains to show the positive side for $r=\rankbyscore{\scoremax}{\desc}$ and fixed $k$.
First, consider the expected intersection size $\peffectf{k}{\cap}(\pi) \defeq \size{T\cap T_0}$.
Using linearity of expectation, we can solve the problem with $k$ calls to the top-$k$ membership algorithm:
\[
\expectation_{\vect u\sim\Pi}[\peffectf{k}{\cap}(\pi)] = 
\expectation_{\vect u\sim\Pi}\big[ \sum_{i \in T_0}  \one_{i \in T} \big]
= \sum_{i \in T_0} \expectation_{\vect u\sim\Pi} \big[\one_{i \in T} \big].
\]
By \Cref{thm:max-topk-ptime}, $\expectation_{\vect u\sim\Pi} \big[\one_{i \in T}\big]$ can be computed in polynomial time for fixed $k$,
therefore $\peffectf{k}{\cap}$ can also be computed in polynomial time.

For the expected change $\peffectf{k}{any}$, we consider each tuple $\vect{t} \in T_0$ individually, and use the same approach as in as in the proof of \Cref{thm:max-topk-ptime} to solve the top-$1$ membership problem $\EXP{\rankbyscore{\scoremax}{\desc}}{\peffectf{k}{top}}$ for $\vect{t}$ and a matrix that does not contain the $k-1$ tuples $T_0 \setminus \{\vect{t}\}$.
In particular, we merge the $n-k$ tuples that are not in $T_0$ into two competitor tuples $\vect{t'}$ and $\vect{t''}$ depending on whether they win over $\vect{t}$ in the tie-breaking scheme.
In both cases, tuples $\vect{t}$ and $\vect{t'}$ contain the largest value per column among their respective set.
We then invoke the algorithm $\PREC{\rankbyscore{\scoremax}{\desc}}$ to return the list of weights that result in top-$1$ membership for $\vect{t}$ (recall \Cref{obs:list_of_weights}).
Finally, we intersect the lists of weights for different $\vect{t} \in T_0$ and compute their probability to obtain the desired expectation.
\end{proof}

\section{Hardness of Hamming Distance}\label{sec:appendix:hardness-hamming}

In this section, we show that \cref{lem:md-hard-matrices} also holds for the Hamming distance $\distham(\pi_1, \pi_2) = \sum_{1 \leq i \leq n} \mathbbm 1_{\pi_1(i) \neq \pi_2(i)}$.

\begin{lemma}
    Let $\phi = \bigwedge_{i = 1}^\ell D_i$ be a positive CNF formula with variable set $X_1, \ldots, X_m$ and clauses $D_i = \bigvee_{j = 1}^{r_i} X_{j_i}$. Then, there exist two a matrices $\Mat_1 \in \M^{(2\ell + 1) \times m}$ and $\Mat_2 \in \M^{(2\ell + 2) \times m}$ with entries in $\set{0,1}$ and the following property: For each set $C$ of columns and each ranking function $r \in \set{\rankbyscore{\scoremax}{\asc}, \ranksum, \ranklex}$, we have
    \begin{equation}\label{eqn:diff-ham}
        \distham\big(r(\Mat_2), r(\restrict{\Mat_2}{C})\big) - \distham\big(r(\Mat_1), r(\restrict{\Mat_1}{C})\big) = \begin{cases}
        0,\text{ if } \assign{C}\models \phi \\
        1,\text{ otherwise.}
        \end{cases}
    \end{equation}
\end{lemma}

\begin{proof}
    We show that the matrices $\Mat_1$ and $\Mat_2$ defined in \cref{lem:md-hard-matrices} also fulfill \cref{eqn:diff-ham}. Let us adopt the notation from the proof of \cref{lem:md-hard-matrices}. If $\assign{C} \models \phi$, then we already showed that the all-zero tuples have the same rank in $r(\Mat_i)$ and $r(\restrict{\Mat_i}{C})$ and the other tuples have to same displacement for $\Mat_1$ and $\Mat_2$, so in particular the number of displacements is equal and \cref{eqn:diff-ham} holds for that case.

    If $\assign{C} \nvDash \phi$, then all the tuples corresponding to non-satisfied clauses and the all-zero tuples change their positions while tuples corresponding to satisfied clauses change their position in $\Mat_1$ iff they change their position in $\Mat_2$. Since $\Mat_2$ has one more all-zeros tuple than $\Mat_1$, this shows the claim.
\end{proof}

\end{document}